\UseRawInputEncoding 
\documentclass[11pt,reqno]{amsart}
\usepackage{tikz}
\usepackage{subfig}
\usepackage{epstopdf}
\usepackage{epsfig}
\usepackage{math}
\graphicspath{{./figures/}}
\usepackage{tikz}
\usetikzlibrary{shapes,arrows}
\tikzstyle{decision} = [diamond, draw, fill=blue!20, 
    text width=4.5em, text badly centered, node distance=3cm, inner sep=0pt]
\tikzstyle{block} = [rectangle, draw, fill=blue!20, 
    text width=5em, text centered, rounded corners, minimum height=4em]
\tikzstyle{line} = [draw, -latex']
\tikzstyle{cloud} = [draw, ellipse,fill=red!20, node distance=3cm,
    minimum height=2em]

\newcommand{\ts}{\mathsf{T}}
\newcommand{\rt}{\mathrm{T}}
\newcommand{\rtKL}{\mathrm{TKL}}
\newcommand{\rtJS}{\mathrm{TJS}}
\usetikzlibrary{positioning}
\tikzset{main node/.style={circle,fill=blue!20,draw,minimum size=1cm,inner sep=0pt},  }

\begin{document}
\title[Transport information Bregman divergences]{Transport information Bregman Divergences}
\author[Li]{Wuchen Li}
\email{wuchen@mailbox.sc.edu}
\address{Department of Mathematics, University of South Carolina.}
\keywords{Transport Bregman divergence; Transport KL divergence; Transport Jensen-Shannon divergence.} 
\thanks{W.~Li is supported by the start up funding in University of South Carolina.} 
\begin{abstract}
We study Bregman divergences in probability density space embedded with the $L^2$--Wasserstein metric. Several properties and dualities of transport Bregman divergences are provided. In particular, we derive the  transport Kullback--Leibler (KL) divergence by a
Bregman divergence of negative Boltzmann--Shannon entropy in $L^2$--Wasserstein space. We also derive analytical formulas and generalizations of transport KL divergence for one-dimensional probability densities and Gaussian families.   
\end{abstract}
\maketitle
\section{Introduction}
Bregman divergences between probability densities are crucial in statistical inference, optimization, and image/signal processing with vast applications in AI inference problems and optimizations \cite{CA, Nielsen, BI}. They measure differences between two densities by generalizing $L^2$ (Euclidean) distances. In general, the Bregman divergence is not symmetric and satisfies several duality properties, which are useful in estimation and optimization algorithms. One typical example is the Kullback--Leibler (KL) divergence, which is a Bregman divergence of (negative) Boltzmann-Shannon entropy in $L^2$ space. 

Information geometry \cite{IG, AA, IG2} studies properties of Bregman divergences. It focuses on the Fisher-Rao information metric, a.k.a. the Hessian metric of (negative) Boltzmann--Shannon entropy in $L^2$ space. A known fact is that the Fisher-Rao information metric can be used to construct the KL divergence and its generalizations \cite{AA, FA} with desirable duality properties. 

Recently, optimal transport, a.k.a. Wasserstein distance, introduces the other type of distance functions in probability density space. 
It uses the pushforward mapping functions to measure differences between probability densities \cite{Villani2009_optimal}.
A particular example is the $L^2$--Wasserstein distance, which forms an analog of $L^2$ distance between mapping functions. It also introduces a metric space for probability densities, namely the $L^2$--Wasserstein space \cite{AGS, otto2001}. In this space, the $L^2$--Wasserstein distance shows a particular convexity property towards mapping functions \cite{AGS, MC}. This convexity property nowadays has vast applications in fluid dynamics \cite{ESG, Lafferty, LiG}, inverse problems \cite{YF}, and AI inference problems \cite{AKM, GHY, PC}. 

Natural questions arise. {\em What are Bregman divergences in $L^2$--Wasserstein space? In particular, what is the ``KL divergence''  in $L^2$--Wasserstein space?}

In this paper, we formulate Bregman divergences in $L^2$--Wasserstein space, namely transport Bregman divergences. We study several properties of transport Bregman divergences. In particular, we derive the transport Bregman divergence of (negative) Boltzmann--Shanon entropy. It can be viewed as the KL divergence in $L^2$--Wasserstein space, whose properties, examples, and symmetric generalizations are provided. 

We briefly present the main result. Denote a compact smooth set by $\Omega\subset\mathbb{R}^d$, and let $\mathcal{P}(\Omega)$ be the probability density space supported on $\Omega$. Given a ``convex'' functional $\mathcal{F}\colon\mathcal{P}(\Omega)\rightarrow\mathbb{R}$, define
\begin{equation*}
\mathrm{D}_{\rt,\mathcal{F}}(p\|q)=\mathcal{F}(p)-\mathcal{F}(q)-\int_\Omega \Big(\nabla_x\frac{\delta}{\delta q(x)}\mathcal{F}(q), \nabla_x\Phi_p(x)-x\Big)q(x)dx,
\end{equation*}
where $p$, $q\in \mathcal{P}(\Omega)$, $\frac{\delta}{\delta q(x)}$ is the $L^2$ first variation w.r.t. $q(x)$, and $\nabla_x\Phi_p$ is the optimal transport map function that pushforwards $q$ to $p$, such that
\begin{equation*}
(\nabla_x{\Phi_p})_{\#}q=p.
\end{equation*}
Here we call $\mathrm{D}_{\rt,\mathcal{F}}$ the transport Bregman divergence. If $\mathcal{F}$ is a second moment functional, i.e. $\mathcal{F}(p)=\int_\Omega \|x\|^2 p(x)dx$, then $\mathrm{D}_{\rt, \mathcal{F}}$ forms the $L^2$--Wasserstein distance. If $\mathcal{F}$ is the negative Boltzmann--Shanon entropy, i.e. $\mathcal{F}(p)=\int_\Omega p(x)\log p(x)dx$, then the transport Bregman divergence satisfies 
\begin{equation*}
\mathrm{D}_{\rtKL}(p\|q)=\mathrm{D}_{\rt,\mathcal{F}}(p\|q)=\int_\Omega \Big(\Delta_x\Phi_p(x)-\log\mathrm{det}(\nabla^2_x\Phi_p(x))-d\Big)q(x)dx.
\end{equation*}
We name $\mathrm{D}_{\rtKL}$ the {\em transport KL divergence}. We notice that $\mathrm{D}_{\rtKL}$ has a closed form formula in one dimensional sample space. 
\begin{equation*}
\mathrm{D}_{\rtKL}(p\|q):=\int_0^1\Big(\frac{\nabla_x F_p^{-1}(x)}{\nabla_x F_q^{-1}(x)}-\log\frac{\nabla_x F_p^{-1}(x)}{\nabla_x F_q^{-1}(x)}-1\Big)dx,
\end{equation*}
where $F_p$, $F_q$ are cumulative distribution functions (CDFs) of $p$, $q$, respectively, and $F_p^{-1}$, $F_q^{-1}$ are their inverse CDFs. We remark that $\mathrm{D}_{\rtKL}$ is an Itakura--Saito type divergence in term of mapping functions. 

There are joint works in the literature between optimal transport and information geometry to study Bregman divergences \cite{AKM, Modin1, GHY, MP1, MP2, KM, Wong, WY}. In particular, \cite{GHY, Wong, WY} apply linear programming formulations of optimal transport. They use divergence functions on sample space as ground costs to construct the ones in probability density space. Compared to the above approaches, we define Bregman divergences by Jacobi operators of mapping functions. And our divergence functionals are built from both gradient and Hessian operators of information entropies in $L^2$--Wasserstein space. Besides, our definition inherits ideas from the Wasserstein subdifferential calculus defined in \cite{AGS}. Here we focus on formulations and generalizations towards transport Bregman divergences. 

This paper is organized as follows. In section \ref{sec2}, we review the definition of Bregman divergence in Euclidean space. In section \ref{sec3} and \ref{sec4}, we construct transport Bregman divergences and study their properties. In section \ref{section4}, we study the transport KL divergence and its symmetric generalizations. Several analytical formulas of transport KL divergences for one-dimensional probability densities and Gaussian families are provided.  
\section{Bregman divergences in Euclidean space}\label{sec2}
In this section, we briefly recall the definition of Bregman divergences in Euclidean space. Bregman divergences generalize Euclidean distances as follows. 
\begin{definition}[Bregman divergence]\label{def1}
Denote a closed convex set $\Omega\subset\mathbb{R}^d$. Let $(\cdot, \cdot)$ be a Euclidean inner product and denote the Euclidean norm by $\|\cdot\|$. Let $\psi\colon \Omega\rightarrow\mathbb{R}$ be a smooth strictly convex function. The Bregman divergence $\mathrm{D}_\psi\colon \Omega\times\Omega\rightarrow\mathbb{R}$ is defined by 
\begin{equation*}
\mathrm{D}_{\psi}(y\|x)=\psi(y)-\psi(x)-(\nabla\psi(x), y-x),\quad\textrm{for any $x,y\in\Omega$.}
\end{equation*}
\end{definition}
We present several examples of Bregman divergences. 
\begin{itemize}
\item[(i)] Let $\Omega=\mathbb{R}^1$ and $\psi(z)=z^2$, $z\in \Omega$. Then  
\begin{equation*}
\mathrm{D}_\psi(y\|x)=y^2-x^2-2x(y-x)=(y-x)^2.
\end{equation*}
Here $\mathrm{D}_\psi$ forms the Euclidean distance. 
\item[(ii)] Let $\Omega=\mathbb{R}_+$ and $\psi(z)=z\log z$, $z\in \Omega$. Then  
\begin{equation*}
\mathrm{D}_\psi(y\|x)=y\log\frac{y}{x}-(y-x).
\end{equation*}
Here $\mathrm{D}_\psi$ leads to the KL divergence in $\mathbb{R}_+$.  
\item[(iii)]  Let $\Omega=\mathbb{R}_+$ and $\psi(z)=-\log z$, $z\in\Omega$.  Then 
\begin{equation}\label{IS}
\mathrm{D}_\psi(y\|x)
=\frac{y}{x}-\log\frac{y}{x}-1.
\end{equation}
Here $\mathrm{D}_\psi$ is known as the Itakura--Saito divergence in $\mathbb{R}_+$.
\end{itemize}
There are several properties of Bregman divergences. 
\begin{itemize}
\item Nonnegativity: $\mathrm{D}_\psi(y\|x)\geq 0$;
\item Hessian metric: Consider a Taylor expansion as follows. Denote $\Delta x\in \mathbb{R}^d$, then  
\begin{equation*}
\mathrm{D}_\psi(x+\Delta x\|x)=\frac{1}{2}\Delta x^{\ts}\nabla^2\psi(x)\Delta x+o(\|\Delta x\|^2),
\end{equation*}
where $\nabla^2\psi$ is the Hessian operator of $\psi$ w.r.t. the Euclidean metric. If $\psi(z)=z\log z$, then $\nabla^2\psi=\frac{1}{z}$ is known as the Fisher-Rao information metric;
\item Asymmetry: In general, $\mathrm{D}_\psi$ is not necessary symmetric w.r.t. $x$ and $y$, i.e. $\mathrm{D}_\psi(y\|x)\neq \mathrm{D}_\psi(x\|y)$. For this reason, we call $\mathrm{D}_\psi$ the ``divergence'' function instead of a distance function;
\item Convexity in the first variable: $\mathrm{D}_\psi(y\|x)$ is always convex w.r.t. $y$, not necessary w.r.t. $x$. The Itakura--Saito divergence \eqref{IS} is an example.  
\item Duality:  Denote the conjugate (dual) function of $\psi$ by $\psi^*(x^*)=\sup_{x\in\Omega}~(x,x^*)-\psi(x)$. 
Then 
\begin{equation*}
 D_{\psi^*}(x^{*}\|y^*)=\mathrm{D}_\psi(y\|x).
\end{equation*}
Here $x^*=\nabla\psi(x)$, $y^*=\nabla\psi(y)$ are dual points corresponding to $x$ and $y$.
\end{itemize}
In practice, Bregman divergences have been extensively studied in probability density space embedded with the $L^2$ metric, which have vast applications in statistics, optimization and AI. In this paper, instead of using the $L^2$ metric, we formulate and study Bregman divergences w.r.t. the $L^2$--Wasserstein metric. 

\section{Bregman divergences in $L^2$--Wasserstein space}\label{sec3}
In this section, we define the Bregman divergence in $L^2$--Wasserstein space. Several concrete examples are provided.  
\subsection{Review of $L^2$--Wasserstein space}
We briefly recall some facts in $L^2$--Wasserstein space \cite{AGS, Villani2009_optimal}. Denote $\Omega$ by a $d$--dimensional compact convex set. E.g., $\Omega=\mathbb{T}^d$, which is a $d$-dimensional torus. Denote the smooth probability density space by 
\begin{equation*}
\mathcal{P}(\Omega)=\Big\{p\in C^{\infty}(\Omega)\colon \int_\Omega p(x)dx=1,~p(x)\geq 0\Big\}.
\end{equation*}
Given $p$, $q\in\mathcal{P}(\Omega)$, the $L^2$--Wasserstein distance is defined by
\begin{equation*}
\mathrm{Dist}_{\mathrm{T}}(p,q)^2=\inf_{T\colon \Omega\rightarrow\Omega}\Big\{\int_\Omega\|T(x)-x\|^2q(x)dx\colon T_{\#}q(x)=p(x)\Big\},
\end{equation*}
where the infimum is among all differemorphisms $T$ that pushforward $q$ to $p$. Here, $_\#$ represents a pushforward operator, such that 
\begin{equation*}
p(T(x))\mathrm{det}(\nabla_xT(x))=q(x). 
\end{equation*}
The optimality condition for the map function $T$ can be formulated below.  
\begin{definition}[Transport coordinates]
Given $p$, $q\in \mathcal{P}(\Omega)$, suppose that there exists strictly convex functions $\Phi_p$, $\Phi_q\in C^{\infty}(\Omega)$, such that 
\begin{equation}\label{nc}
T(x)=\nabla_x\Phi_p(x)\quad\textrm{and}\quad \Phi_q(x)=\frac{\|x\|^2}{2}.
\end{equation}
In this case, the $L^2$--Wasserstein distance can be formulated by
\begin{equation*}
\mathrm{Dist}_{\mathrm{T}}(p,q)=\sqrt{\int_\Omega\|\nabla_x\Phi_p(x)-\nabla_x\Phi_q(x)\|^2q(x)dx}.
\end{equation*}
From now on, we always use $\nabla_x\Phi_p$, $\nabla_x\Phi_q$ to represent functions $T$, $x$, respectively. And we call $\Phi_p$, $\Phi_q$ the {transport coordinates}. 
\end{definition}

There is also a linear programming reformulation of $L^2$--Wasserstein distance. Denote the optimal joint probability density function for densities $p$ and $q$ by 
\begin{equation}\label{otpi}
\pi=(\nabla\Phi_q\times \nabla\Phi_p)_\#q=(\textrm{id}\times \nabla\Phi_p)_\#q,
\end{equation}
where 
\begin{equation*}
\int_\Omega \pi(x,y)dy=q(x), \quad\int_\Omega\pi(x,y)dx=p(y).
\end{equation*}
In this sense, the $L^2$--Wasserstein distance in term of the optimal joint density function satisfies 
\begin{equation*}
\mathrm{Dist}_{\mathrm{T}}(p,q)^2=\int_\Omega\int_\Omega \|y-x\|^2\pi(x,y)dxdy.
\end{equation*}

In addition, the $L^2$--Wasserstein distance also introduces a metric in $\mathcal{P}(\Omega)$.  
Denote the tangent space at $q\in \mathcal{P}(\Omega)$ by 
\begin{equation*}
T_q\mathcal{P}(\Omega)=\{\sigma\in C^{\infty}(\Omega)\colon \int_\Omega \sigma(x) dx=0\}.
\end{equation*}
Define an inner product $g_{\mathrm{T}}(q)\colon T_q\mathcal{P}(\Omega)\times T_q\mathcal{P}(\Omega)\rightarrow\mathbb{R}$, such that
\begin{equation*}
g_{\mathrm{T}}(q)(\sigma,\sigma)=\int_\Omega (\nabla_x\Phi(x), \nabla_x\Phi(x))q(x)dx,
\end{equation*}
where $\Phi, \sigma \in C^{\infty}(\Omega)$ satisfy 
\begin{equation*}
\sigma(x)=-\nabla_x\cdot(q(x)\nabla_x\Phi(x)). 
\end{equation*}
In literature, $(\mathcal{P}(\Omega), g_{\mathrm{T}})$ is often called the $L^2$--Wasserstein space. In this paper, we introduce Bregman divergences in $(\mathcal{P}(\Omega), g_{\mathrm{T}})$. To do so, we review several useful facts . 
\begin{proposition}[Facts in $L^2$--Wasserstein space \cite{AGS, Villani2009_optimal}]\label{prop4}
The following facts hold. 
 \begin{itemize}
\item[(i)] The $L^2$ gradient of functional $\mathrm{Dist}_{\mathrm{T}}(p,q)^2$ w.r.t. $q(x)$ satisfies 
\begin{equation*}
\frac{1}{2}\frac{\delta}{\delta q(x)}\mathrm{Dist}_{\rt}(p,q)^2=\Phi_p(x)-\Phi_q(x),
\end{equation*}
where $\Phi_p$, $\Phi_q$ are defined in \eqref{nc}. Here $\frac{\delta}{\delta q(x)}$ is the $L^2$ first variation operator w.r.t density $q$ at $x\in\Omega$. 
\item[(ii)] Denote a smooth functional $\mathcal{F}\colon\mathcal{P}(\Omega)\rightarrow\mathbb{R}$. The gradient operator of functional $\mathcal{F}$ in $(\mathcal{P}(\Omega), g_{\mathrm{T}})$ satisfies 
\begin{equation*}
\textrm{grad}_{\mathrm{T}}\mathcal{F}(q)(x)=-\nabla_x\cdot(q(x)\nabla_x \frac{\delta}{\delta q(x)}\mathcal{F}(q))\in T_q\mathcal{P}(\Omega).
\end{equation*}
\item[(iii)] The Hessian operator of functional $\mathcal{F}$ in $(\mathcal{P}(\Omega), g_{\mathrm{T}})$ satisfies 
\begin{equation}\label{Hessian}
\begin{split}
\mathrm{Hess}_{\mathrm{T}}\mathcal{F}(q)(\sigma,\sigma)=&\int_\Omega\int_\Omega \nabla^2_{xy}\frac{\delta^2}{\delta q(x)\delta q(y)}\mathcal{F}(q)(\nabla\Phi(x), \nabla\Phi(y))q(x)q(y)dxdy\\ 
&+\int_\Omega \nabla^2_x\frac{\delta}{\delta q(x)}\mathcal{F}(q)(\nabla\Phi(x), \nabla\Phi(x))q(x)dx,
\end{split}
\end{equation}
where $\Phi$, $\sigma\in C^{\infty}(\Omega)$ satisfy $\sigma(x)=-\nabla_x\cdot(q(x)\nabla_x\Phi(x))$, $\frac{\delta^2}{\delta q(x)\delta q(y)}$ is the $L^2$ second variation operator w.r.t. $q(x)$, $q(y)$ respectively,  $\nabla^2_{xy}$ is the second differential operator in $\Omega$ w.r.t. $x$, $y$ respectively, and $\nabla_x^2$ is the Hessian operator in $\Omega$ w.r.t. $x$; see details in \cite{LiG}; 
\item[(iv)] We say that functional $\mathcal{F}$ is $\lambda$--geodesic convex in $(\mathcal{P}(\Omega), g_{\mathrm{T}})$, if {for any $\sigma\in T_q\mathcal{P}(\Omega)$ and $q\in\mathcal{P}(\Omega)$,} there exists a constant $\lambda$, such that 
\begin{equation*}
\mathrm{Hess}_{\mathrm{T}}\mathcal{F}(q)(\sigma,\sigma)\geq \lambda g_{\mathrm{T}}(q)(\sigma,\sigma). 
\end{equation*}
In other words, for any $p, q\in\mathcal{P}(\Omega)$ with $\Phi_p, \Phi_q$ defined in \eqref{nc}, then for any $t\in [0,1]$, 
\begin{equation*}
\mathcal{F}((t\nabla\Phi_p+(1-t)\nabla\Phi_q)_\#q)\leq t \mathcal{F}(p)+(1-t)\mathcal{F}(q)-\lambda\frac{t(1-t)}{2}\mathrm{Dist}_{\rt}(p,q)^2.
\end{equation*}
In literature \cite{MC} and \cite[Definition 16.5]{Villani2009_optimal}, the geodesic convexity in $(\mathcal{P}(\Omega), g_{\mathrm{T}})$ is named the displacement convexity. 
\end{itemize}
\end{proposition}
\subsection{Transport Bregman divergence}
We are now ready to state the main result of this paper. We define Bregman divergences in $L^2$--Wasserstein space. 
\begin{definition}[Transport Bregman divergence]\label{def4}
Let $\mathcal{F}\colon \mathcal{P}(\Omega)\rightarrow\mathbb{R}$ be a smooth strictly displacement convex functional. Define $\mathrm{D}_{\rt,\mathcal{F}}\colon \Omega\times\Omega\rightarrow\mathbb{R}$ by 
\begin{equation}\label{TB}
\mathrm{D}_{\rt,\mathcal{F}}(p\|q)=\mathcal{F}(p)-\mathcal{F}(q)-\int_\Omega \Big(\nabla_x\frac{\delta}{\delta q(x)}\mathcal{F}(q), T(x)-x\Big)q(x)dx,
\end{equation}
where $T(x)$ is the optimal transport map function \eqref{nc} from $q$ to $p$, such that
\begin{equation*}
T_{\#}q=p\quad\textrm{and}\quad T(x)=\nabla_x\Phi_p(x).
\end{equation*}
In other words, 
\begin{equation*}
\mathrm{D}_{\rt,\mathcal{F}}(p\|q)=\mathcal{F}(p)-\mathcal{F}(q)-\int_\Omega \Big(\nabla_x\frac{\delta}{\delta q(x)}\mathcal{F}(q), \nabla_x\Phi_p(x)-\nabla_x\Phi_q(x)\Big)q(x)dx,\end{equation*}
where $\Phi_p$, $\Phi_q$ are defined in \eqref{nc}.
We call $\mathrm{D}_{\rt, \mathcal{F}}$ the {\em transport Bregman divergence}. 
\end{definition}
 The following proposition describes that functional \eqref{TB} is a generalization of Euclidean Bregman divergence in Definition \ref{def1}. 
\begin{proposition}\label{prop5}
Functional $\mathrm{D}_{\rt,\mathcal{F}}$ satisfies the following equality
\begin{equation*}
\begin{split}
\mathrm{D}_{\rt,\mathcal{F}}(p\|q)=&\mathcal{F}(p)-\mathcal{F}(q)-\frac{1}{2}\int_\Omega \textrm{grad}_{\mathrm{T}}\mathcal{F}(q)(x)\cdot \frac{\delta}{\delta q(x)}\mathrm{Dist}_{\rt}(p,q)^2 dx.
\end{split}
\end{equation*}
\end{proposition}
\begin{proof}
The proof follows from the properties of the $L^2$--Wasserstein space. From Proposition \ref{prop4} (i), (ii), we have
\begin{equation}
\begin{split}
&\frac{1}{2}\int_\Omega \textrm{grad}_{\mathrm{T}}\mathcal{F}(q)(x)\cdot \frac{\delta}{\delta q(x)}\mathrm{Dist}_{\rt}(p,q)^2 dx\\=&-\int_\Omega \nabla_x\cdot(q(x)\nabla_x\frac{\delta}{\delta q(x)}\mathcal{F}(q))\cdot (\Phi_p(x)-\Phi_q(x)) dx\\
=&\int_\Omega \Big(\nabla_x\frac{\delta}{\delta q(x)}\mathcal{F}(q),\nabla_x\Phi_p(x)-\nabla_x\Phi_q(x)\Big)q(x) dx\\
=&\int_\Omega \Big(\nabla_x\frac{\delta}{\delta q(x)}\mathcal{F}(q),T(x)-x\Big)q(x) dx,
\end{split}
\end{equation}
where the second equality holds by the integration by parts formula. And we also apply the fact that $T(x)=\nabla_x\Phi_p(x)$ is the optimal mapping function. 
\end{proof}

We also formulate transport Bregman divergences by the optimal joint density function. 
\begin{proposition}\label{BJD}
Denote $\pi$ by \eqref{otpi}. Then  
\begin{equation*}
\mathrm{D}_{\rt,\mathcal{F}}(p\|q)=\mathcal{F}(p)-\mathcal{F}(q)-\int_\Omega \Big(\nabla_x\frac{\delta}{\delta q(x)}\mathcal{F}(q), y-x\Big)\pi(x,y)dxdy.
\end{equation*}
\end{proposition}
\begin{proof}
We represent the optimal mapping function by 
\begin{equation*}
T(x)=\nabla_x\Phi_p(x)=\int_\Omega y\pi(x,y)dy.
\end{equation*}
Hence 
\begin{equation*}
\int_\Omega\Big(\nabla_x\frac{\delta}{\delta q(x)}\mathcal{F}(q),T(x)-x\Big)q(x)dx=\int_{\Omega}\int_{\Omega}\Big(\nabla_x\frac{\delta}{\delta q(x)}\mathcal{F}(q), y-x\Big)\pi(x,y)dxdy.
\end{equation*}
\end{proof}
\begin{remark}[Connections with Euclidean Bregman divergences]
As shown in proposition \ref{prop5}, we replace the gradient operators of both functional and distance in Euclidean space to the corresponding ones in $L^2$--Wasserstein space. 
\end{remark}
\begin{remark}
Our formulations of Bregman divergences connects with the sub-differential calculus proposed in \cite[Chapter 10]{AGS}. One can also view the transport Bregman divergence in Proposition \ref{BJD} as a definition for transport Bregman divergences. Here the optimal map function $T=\nabla_x\Phi_p$ is not necessary required to be smooth. In this paper, we focus on formulations of transport Bregman divergences, and leave their analytical properties in future works. 
\end{remark}
\subsection{Formulations}
We first demonstrate several examples of transport Bregman divergences.   
\begin{proposition}[Mapping formulations]\label{prop7}
Transport Bregman divergences have the following formulations in term of mapping functions $\Phi_p$, $\Phi_q$ defined in \eqref{nc}.
\begin{itemize}
\item[(i)] Consider a linear energy by 
\begin{equation*}
\mathcal{V}(p)=\int_\Omega V(x)p(x)dx,
\end{equation*}
where the linear potential function $V\in C^{\infty}(\Omega)$ is strictly convex in $\mathbb{R}^d$. Then 
\begin{equation}\label{linear}
\begin{split}
\mathrm{D}_{\rt,\mathcal{V}}(p\|q)=&\int_\Omega \mathrm{D}_{V}(\nabla_x\Phi_p(x)\|\nabla_x\Phi_q(x))q(x)dx,
\end{split}
\end{equation}
where $\mathrm{D}_{V}\colon\Omega\times\Omega\rightarrow\mathbb{R}$ is a Euclidean Bregman divergence of $V$ defined by
\begin{equation*}
\mathrm{D}_V(z_1\|z_2)=V(z_1)-V(z_2)-\nabla V(z_2)\cdot (z_1-z_2),\quad\textrm{for any $z_1,z_2\in\Omega$}.
\end{equation*}
\item[(ii)] Consider an interaction energy by
\begin{equation*}
\mathcal{W}(p)=\frac{1}{2}\int_\Omega\int_\Omega W(x,\tilde x)p(x)p(\tilde x)dxd\tilde x,
\end{equation*}
where the interaction kernel potential function is $W(x,\tilde x)=W(\tilde x,x)\in C^{\infty}(\Omega\times\Omega)$. Assume $W(x,\tilde x)=\tilde W(x-\tilde x)$, where $\tilde W(z)$ is a convex function of $z$ in $\Omega$ with $\tilde W(z)=\tilde W(-z)$. Then
\begin{equation}\label{interact}
\begin{split}
\mathrm{D}_{\rt,\mathcal{W}}(p\|q)=\frac{1}{2}\int_\Omega\int_\Omega \mathrm{D}_{\tilde W}\big(\nabla_x\Phi_p(x)-\nabla\Phi_p(\tilde x)\|\nabla\Phi_q(x)-\nabla\Phi_q(\tilde x)\big)q(x)q(\tilde x)dxd\tilde x,
\end{split}
\end{equation}
where $\mathrm{D}_{\tilde W}\colon \Omega\times\Omega\rightarrow \mathbb{R}$ is a Euclidean Bregman divergence of $\tilde W$ defined by
\begin{equation*}
\mathrm{D}_{\tilde W}(z_1\|z_2)=\tilde W(z_1)-\tilde W(z_2)-\nabla \tilde W(z_2)\cdot (z_1-z_2),\quad\textrm{for any $z_1,z_2\in\Omega$}.
\end{equation*}
\item[(iii)] Consider a negative entropy by
\begin{equation*}
\mathcal{U}(p)=\int_\Omega U(p(x))dx,
\end{equation*}
where the entropy potential $U\colon \Omega\rightarrow\mathbb{R}$ is second differentiable and convex. Then
\begin{equation*}
\mathrm{D}_{\rt,\mathcal{U}}(p\|q)=\int_\Omega \mathrm{D}_{\hat U}\big({\nabla_x^2\Phi_p(x)}\|\nabla_x^2\Phi_q(x)\big)q(x)dx,
\end{equation*}
where $\mathrm{D}_{\hat U}\colon \mathbb{R}^{d\times d}\times \mathbb{R}^{d\times d}\rightarrow\mathbb{R}$ is a matrix Bregman divergence function. Denote function $\hat U\colon \mathbb{R}_+\times\mathbb{R}^{d\times d}\rightarrow \mathbb{R}$ by 
\begin{equation*}
\hat U(q,A)=U(\frac{q}{\mathrm{det}(A)})\frac{\mathrm{det}(A)}{q},
\end{equation*}
where $q\in\mathbb{R}_+$ is the given reference density, and    
\begin{equation*}
\mathrm{D}_{\hat U}(A\|B)=\hat U(q,A)-\hat U(q, B)-\mathrm{tr}\Big(\nabla_B \hat U(q, B)\cdot (A-B)\Big), 
\end{equation*}
for any $A$, $B\in\mathbb{R}^{d\times d}$ and $\nabla_B$ is the Frechet derivative of a symmetric matrix $B$.
In details, 
\begin{equation}\label{entropy}
\begin{split}
\mathrm{D}_{\rt,\mathcal{U}}(p\|q)
=\int_\Omega &\Big\{U(\frac{q(x)}{\mathrm{det}(\nabla^2_x \Phi_p(x))})\mathrm{det}(\nabla^2_x \Phi_p(x))-U(q(x))\\
&+\mathrm{tr}(\nabla^2_x\Phi_p(x)-\mathbb{I})\big[U'(q(x))q(x)-U(q(x))\big]\Big\}dx,
\end{split}
\end{equation}
where $\mathbb{I}\in\mathbb{R}^{d\times d}$ is an identity matrix.
\end{itemize}
\end{proposition}
\begin{proposition}[Joint density formulations]\label{prop8}
Transport Bregman divergences have the following formulations by the joint density function $\pi$ defined in \eqref{otpi}.
\begin{itemize}
\item[(i)] Formula \eqref{linear} leads to 
\begin{equation*}
\mathrm{D}_{\rt,\mathcal{V}}(p\|q)=\int_\Omega\int_\Omega \mathrm{D}_V(y\|x)\pi(x,y)dxdy. 
\end{equation*}
\item[(ii)] Formula \eqref{interact} leads to 
\begin{equation*}
\mathrm{D}_{\rt,\mathcal{W}}(p\|q)=\frac{1}{2}\int_{\Omega}\int_{\Omega}\int_{\Omega}\int_{\Omega} \mathrm{D}_{\tilde W}(y-\tilde y\|x-\tilde x)\pi(x,y)\pi(\tilde x,\tilde y)dxdyd\tilde xd\tilde y.
\end{equation*}
\item[(iii)] Formula \eqref{entropy} leads to
\begin{equation*}
\begin{split}
\mathrm{D}_{\rt,\mathcal{U}}(p\|q)=&\quad\int_\Omega U(\int_\Omega \pi(x,y)dx)dy-\int_\Omega U(\int_\Omega \pi(x,y)dy)dx\\
&+\int_\Omega\int_\Omega (y\cdot\nabla_x\frac{\pi(x,y)}{q(x)}-d)\bar U(\int_\Omega \pi(x,z)dz)dxdy,
\end{split}
\end{equation*}
where we denote 
\begin{equation*}
\bar U(z):=zU'(z)-U(z). 
\end{equation*}
\end{itemize}
\end{proposition}
\begin{remark}[Comparisons with Wasserstein-Bregman divergence]
We notice that formulation (i) is similar but different with divergence functionals defined in \cite{GHY, Wong}. In our formulation, joint density $\pi$ is the optimal one from the $L^2$--Wasserstein space; while in \cite{GHY} or \cite{Wong}, joint density $\pi$ solves the related linear programming problem w.r.t. divergence type ground costs. 
\end{remark}
\begin{remark}
We remark that formulation (ii) leads to a quadratic programming problem with interaction potential divergence functions $\mathrm{D}_{\tilde W}$. 
\end{remark}
\begin{remark}
In information geometry \cite{IG} and information theory \cite{CoverThomas1991_elements}, the entropy functional $\mathcal{U}(p)$ is often named information entropy. Because of this reason, we call divergence functional \eqref{entropy} the {\em transport information Bregman divergence}. 
\end{remark}

Besides, we derive several closed form solutions of transport Bregman divergences in one dimensional space. 
\begin{proposition}[One dimensional closed form solutions]\label{prop9}
Let $\Omega=[0,1]$ or $\mathbb{T}^1$. Denote $F_p$, $F_q$ as the cumulative distribution of densities $p$, $q$, respectively, and let $F_p^{-1}$, $F_q^{-1}$ be
their inverse functions. We have the following closed formulas of transport Bregman divergences. 
\begin{itemize}
\item[(i)] The transport Bregman divergence of linear energy $\mathcal{V}$ \eqref{linear} satisfies  
\begin{equation*}
\begin{split}
\mathrm{D}_{\rt,\mathcal{V}}(p\|q)=\int_0^1 \mathrm{D}_{V}(F^{-1}_p(x)\|F^{-1}_q(x))dx.
\end{split}
\end{equation*}
\item[(ii)] The transport Bregman divergence of interaction energy $\mathcal{W}$ \eqref{interact} satisfies
\begin{equation*}
\begin{split}
\mathrm{D}_{\rt,\mathcal{W}}(p\|q)=&\frac{1}{2}\int_0^1\int_0^1 \mathrm{D}_{\tilde W}(F^{-1}_p(x)-F^{-1}_p(\tilde x)\|F^{-1}_q(x)-F^{-1}_q(\tilde x))dxd\tilde x.
\end{split}
\end{equation*}
\item[(iii)] The transport Bregman divergence of negative entropy $\mathcal{U}$ \eqref{entropy} satisfies 
\begin{equation}\label{ent1d}
\begin{split}
\mathrm{D}_{\rt,\mathcal{U}}(p\|q)=\int_0^1 \mathrm{D}_{\tilde U}(\nabla_xF_p^{-1}(x)\|\nabla_xF_q^{-1}(x))dx,
\end{split}
\end{equation}
where we denote 
\begin{equation*}
\tilde U(z)=zU(\frac{1}{z}),
\end{equation*}
and $\mathrm{D}_{\tilde U}$ is the Euclidean Bregman divergence of $\tilde U$ defined by 
\begin{equation*}
\mathrm{D}_{\tilde U}(z_1\|z_2)=\tilde U(z_1)-\tilde U(z_2)-\nabla_z\tilde U(z_2)\cdot (z_1-z_2).
\end{equation*}
\end{itemize}
\end{proposition}
\subsection{Examples}
In this subsection, we present several examples of transport Bregman divergences. 
\begin{example}[Linear energy]
We remark that the transport Bregman divergence of linear energy \eqref{linear} is the expectation of a ``linear Bregman divergence potential''. Several examples are given below.  
\begin{itemize}
\item[(i)] If $\mathcal{V}(\rho)=\int_\Omega \|x\|^2 p(x)dx$, then
\begin{equation*}
\begin{split}
\mathrm{D}_{\rt,\mathcal{V}}(p\|q)=&\int_\Omega \Big(\|T(x)\|^2-\|x\|^2-2(T(x)-x, x)\Big)q(x)dx\\
=&\int_\Omega\|T(x)-x\|^2q(x)dx\\
=&\int_\Omega \|\nabla_x\Phi_p(x)-\nabla_x\Phi_q(x)\|^2q(x)dx\\
=&\int_\Omega\int_\Omega \|y-x\|^2\pi(x,y)dxdy\\
=&\mathrm{Dist}_{\rt}(p,q)^2.
\end{split}
\end{equation*}
In this case, we observe that when $\mathcal{V}(p)=\int_\Omega \|x\|^2 p(x)dx$ is the second moment functional, the transport Bregman divergence forms the $L^2$--Wasserstein distance. In other words, we can view $L^2$--Wasserstein distance as the \textbf{transport Bregman divergence of a second moment functional}. If $d=1$, then 
\begin{equation*}
\mathrm{D}_{\rt,\mathcal{V}}(p\|q)=\int_0^1 \|F^{-1}_p(x)-F^{-1}_q(x)\|^2dx.
\end{equation*}
This is a known closed form solution for $L^2$--Wasserstein distance \cite{AGS}. 
\item[(ii)] If $\Omega=[0,1]$ and $\mathcal{V}(p)=\int_\Omega x\log x p(x)dx$, then  
\begin{equation*}
\begin{split}
\mathrm{D}_{\rt,\mathcal{V}}(p\|q)
=&\int_\Omega\Big(\nabla_x\Phi_p(x)\log\frac{\nabla_x\Phi_p(x)}{\nabla_x\Phi_q(x)}-\big(\nabla_x\Phi_p(x)-\nabla_x\Phi_q(x)\big)\Big)q(x)dx\\
=&\int_\Omega\int_\Omega \Big(y\log\frac{y}{x}-(y-x)\Big)\pi(x,y)dxdy\\
=&\int_0^1\Big(F_p^{-1}(x)\log\frac{F_p^{-1}(x)}{F_q^{-1}(x)}-\big(F_p^{-1}(x)-F^{-1}_q(x)\big)\Big)dx.
\end{split}
\end{equation*}
\item[(iii)] If $\Omega=[0,1]$ and $\mathcal{V}(p)=-\int_\Omega\log x p(x)dx$, then  
\begin{equation*}
\begin{split}
\mathrm{D}_{\rt,\mathcal{V}}(p\|q)=&\int_\Omega\Big(\frac{\nabla_x\Phi_p(x)}{\nabla_x\Phi_q(x)}-\log\frac{\nabla_x\Phi_p(x)}{\nabla_x\Phi_q(x)}-1\Big)q(x)dx\\
=&\int_\Omega\int_\Omega \Big(\frac{y}{x}-\log\frac{y}{x}-1\Big)\pi(x,y)dxdy\\
=&\int_0^1\Big(\frac{F_p^{-1}(x)}{F_q^{-1}(x)}-\log\frac{F_p^{-1}(x)}{F_q^{-1}(x)}-1\Big)dx.
\end{split}
\end{equation*}
\end{itemize}
\end{example}

\begin{example}[Interaction energy]
We remark that the transport Bregman divergence of interaction energy \eqref{interact} is the expectation of an ``interaction Bregman divergence potential''. Two examples are given below. 
\begin{itemize}
\item[(i)] If $\mathcal{W}(p)=\int_\Omega\int_\Omega\|x-\tilde x\|^2p(x)p(\tilde x)dxd\tilde x$, then 
\begin{equation*}
\begin{split}
\mathrm{D}_{\rt,\mathcal{W}}(p\|q)=&\int_\Omega\int_\Omega \|(\nabla \Phi_p(x)-\nabla\Phi_p(\tilde x))-(\nabla\Phi_q(x)-\nabla\Phi_q(\tilde x))\|^2q(x)q(\tilde x)dxd\tilde x\\
=&\int_\Omega\int_\Omega \|(y-\tilde y)-(x-\tilde x)\|^2\pi(x,y)\pi(\tilde x, \tilde y)dxdyd\tilde xd\tilde y.
\end{split}
\end{equation*}
If $d=1$, then  
\begin{equation*}
\mathrm{D}_{\rt,\mathcal{W}}(p\|q)=\int_0^1\int_0^1 \|(F_p^{-1}(x)-F_p^{-1}(\tilde x))-(F_q^{-1}(x)-F_q^{-1}(\tilde x))\|^2dxd\tilde x.
\end{equation*}
\item[(ii)] If $\mathcal{W}(p)=\int_\Omega\int_\Omega\log\frac{1}{\|x-\tilde x\|}p(x)p(\tilde x)dxd\tilde x$, then 
\begin{equation*}
\begin{split}
\mathrm{D}_{\rt,\mathcal{W}}(p\|q)=&\int_\Omega\int_\Omega \Big(\frac{\|\nabla_x\Phi_p(x)-\nabla\Phi_p(\tilde x)\|}{\|\nabla\Phi_q(x)-\nabla\Phi_q(\tilde x)\|}-\log\frac{\|\nabla_x\Phi_p(x)-\nabla\Phi_p(\tilde x)\|}{\|\nabla\Phi_q(x)-\nabla\Phi_q(\tilde x)\|}-1\Big)q(x)q(\tilde x)dxd\tilde x\\
=&\int_\Omega\int_\Omega\int_\Omega\int_\Omega \Big(\frac{\|y-\tilde y\|}{\|x-\tilde x\|}-\log\frac{\|y-\tilde y\|}{\|x-\tilde x\|}-1\Big)\pi(x, y)\pi(\tilde x, \tilde y)dxdyd\tilde xd\tilde y.
\end{split}
\end{equation*}
If $d=1$, then  
\begin{equation*}
\mathrm{D}_{\rt,\mathcal{W}}(p\|q)=\int_0^1\int_0^1 \Big(\frac{\|F_p^{-1}(x)-F_p^{-1}(\tilde x)\|}{\|F_q^{-1}(x)-F_q^{-1}(\tilde x)\|}-\log\frac{\|F_p^{-1}(x)-F_p^{-1}(\tilde x)\|}{\|F_q^{-1}(x)-F_q^{-1}(\tilde x)\|}-1\Big)dxd\tilde x.
\end{equation*}
\end{itemize}
\end{example}
\begin{example}[Negative Entropy]\label{ent}
We remark that the transport Bregman divergence of entropy \eqref{entropy} is the expectation of a matrix Bregman divergence for mapping Jacobi operators.
Two examples are given below. 
\begin{itemize}
\item[(i)] Let $\mathcal{U}(p)=\int_\Omega p(x)\log p(x)dx$.  
Then the transport Bregman divergence forms 
\begin{equation*}
\begin{split}
\mathrm{D}_{\rt,\mathcal{U}}(p\|q)
=&\int_\Omega \Big(\Delta_x\Phi_p(x)-\log\mathrm{det}(\nabla^2_x\Phi_p(x))-d\Big)q(x)dx\\
=&\int_\Omega \Big(p(x)\log p(x)-q(x)\log q(x)\Big)dx +\int_\Omega\int_\Omega (y\cdot\nabla_x\frac{\pi(x,y)}{q(x)}-d)q(x)dxdy.
\end{split}
\end{equation*}
If $d=1$, then
\begin{equation}\label{ED}
\mathrm{D}_{\rt,\mathcal{U}}(p\|q)=\int_0^1\Big(\frac{\nabla_xF_p^{-1}(x)}{\nabla_xF_q^{-1}(x)}-\log \frac{\nabla_xF_p^{-1}(x)}{\nabla_xF_q^{-1}(x)}-1\Big)dx.
\end{equation}
\item[(ii)] 
Let $\mathcal{U}(p)=\int_\Omega \frac{p(x)^2}{2}dx$. Then the transport Bregman divergence forms
\begin{equation*}
\begin{split}
\mathrm{D}_{\rt,\mathcal{U}}(p\|q)=&\frac{1}{2}\int_\Omega\Big(\Delta_x \Phi_p(x)+\frac{1}{\mathrm{det}(\nabla^2_x\Phi_p(x))}-1-d\Big)q(x)^2dx\\
=&\frac{1}{2}\int_\Omega \Big(p(x)^2-q(x)^2\Big)dx+\frac{1}{2}\int_\Omega\int_\Omega (y\cdot\nabla_x\frac{\pi(x,y)}{q(x)}-d)q(x)^2dxdy.
\end{split}
\end{equation*}
If $d=1$, then 
\begin{equation*}
\mathrm{D}_{\rt,\mathcal{U}}(p\|q)=\frac{1}{2}\int_0^1 \Big(\frac{1}{\nabla_xF_p^{-1}(x)}-\frac{1}{\nabla_xF_q^{-1}(x)}\Big)^2\cdot\nabla_xF^{-1}_p(x)dx.
\end{equation*}
\end{itemize}
\end{example}

\section{Properties}\label{sec4}
In this section, we study several properties of transport Bregman divergences. 

\begin{proposition}[Properties]\label{TBP}
Transport Bregman divergence $\mathrm{D}_{\rt,\mathcal{F}}$ has the following properties.
\begin{itemize}
\item[(i)] Non-negativity: Suppose $\mathcal{F}$ is displacement convex, then 
\begin{equation*}
\mathrm{D}_{\rt,\mathcal{F}}(p\|q)\geq 0.
\end{equation*}
Suppose $\mathcal{F}$ is strictly displacement convex, then 
\begin{equation*}
\mathrm{D}_{\rt, \mathcal{F}}(p\|q)=0\quad \textrm{if and only if}\quad \mathrm{Dist}_{\rt}(p,q)=0.
\end{equation*}
\item[(ii)] Transport Hessian metric: Consider a Taylor expansion as follows. Denote $\sigma=-\nabla\cdot(q\nabla\Phi)\in T_q\mathcal{P}(\Omega)$ and $\epsilon\in\mathbb{R}$,
then  
\begin{equation*}
\mathrm{D}_{\rt,\mathcal{F}}((\mathrm{id}+\epsilon \nabla\Phi )_\# q \|q)=\frac{\epsilon^2}{2}\mathrm{Hess}_{\mathrm{T}}\mathcal{F}(q)(\sigma,\sigma)+o(\epsilon^2),
\end{equation*}
where $\mathrm{id}\colon \Omega\rightarrow\Omega$ is the identical map, $\mathrm{id}(x)=x$, and $\mathrm{Hess}_{\mathrm{T}}\mathcal{F}(q)$ is defined in \eqref{Hessian}, which is the Hessian operator of functional $\mathcal{F}$ at $q\in\mathcal{P}(\Omega)$ w.r.t. $L^2$--Wasserstein metric.
\item[(iii)] Linearity: Consider two functionals $\mathcal{F},~\mathcal{G}\colon\mathcal{P}\rightarrow\mathbb{R}$ and $a>0$. Then 
\begin{equation*}
\mathrm{D}_{\rt,\mathcal{F}+a\mathcal{G}}(p\|q)=\mathrm{D}_{\rt,\mathcal{F}}(p\|q)+a\mathrm{D}_{\rt, \mathcal{G}}(p\|q).
\end{equation*}
\item[(iv)] Asymmetry: In general, $\mathrm{D}_{\rt,\mathcal{F}}(p\|q)\neq \mathrm{D}_{\rt,\mathcal{F}}(q\|p)$. One typical example is formula \eqref{ED}, which is an Itakura--Saito divergence in transport coordinates. 
\end{itemize}
\end{proposition}
\begin{proof}
The proof follows from the definition of transport Bregman divergence. (i) Since $\mathcal{F}$ is $\lambda$--geodesic convex in $(\mathcal{P}, g_{\mathrm{T}})$, then 
\begin{equation*}
\mathcal{F}(p)\geq \mathcal{F}(q)+\int_\Omega \textrm{grad}_{\mathrm{T}}\mathcal{F}(q)(x) \cdot \frac{\delta}{\delta q(x)}\mathrm{Dist}_{\rt}(p,q)^2dx+\frac{\lambda}{2}\mathrm{Dist}_{\rt}(p,q)^2,
\end{equation*}
i.e. 
\begin{equation*}
\mathrm{D}_{\rt, \mathcal{F}}(p\|q)\geq \frac{\lambda}{2}\mathrm{Dist}_{\rt}(p,q)^2.
\end{equation*}
Hence if $\lambda\geq 0$, we have $\mathrm{D}_{\rt, \mathcal{F}}(p\|q)\geq 0$. In particular, if $\lambda>0$, then $\mathrm{D}_{\rt,\mathcal{F}}(p\|q)=0$ iff $\mathrm{Dist}_{\mathrm{T}}(p,q)^2=0$. 

(ii) Here the Hessian metric property follows from the definition. Consider a Taylor expansion in $(\mathcal{P}, g_{\mathrm{T}})$ by
\begin{equation*}
\mathcal{F}((\mathrm{id}+\epsilon \nabla\Phi)_\#q)=\mathcal{F}(q)+\epsilon \int_\Omega \textrm{grad}_{\mathrm{T}}\mathcal{F}(q)(x) \cdot \Phi(x)dx+\frac{\epsilon^2}{2}\mathrm{Hess}_{\mathrm{T}}\mathcal{F}(q)(\sigma, \sigma)+o(\epsilon^2).
\end{equation*}
From the definition of transport Bregman divergences, we have
\begin{equation*}
\mathrm{D}_{\rt,\mathcal{F}}((\mathrm{id}+\epsilon \nabla\Phi)_\#q\|q)=\frac{\epsilon^2}{2}\mathrm{Hess}_{\mathrm{T}}\mathcal{F}(q)(\sigma, \sigma)+o(\epsilon^2).
\end{equation*}

(iii) Here, the linearity follows from the linearity of gradient operator in $L^2$--Wasserstein metric. Notice that for any $a>0$, we have
\begin{equation*}
\begin{split}
\textrm{grad}_{\mathrm{T}}\mathcal{F}(q)+a\textrm{grad}_{\mathrm{T}}\mathcal{G}(q)
=&-\nabla\cdot\Big(q(x)\nabla_x\frac{\delta}{\delta q(x)}\mathcal{F}(q)\Big)-a\nabla\cdot\Big(q(x)\nabla_x\frac{\delta}{\delta q(x)}\mathcal{G}(q)\Big)\\
=&-\nabla\cdot\Big(q(x)\nabla_x\frac{\delta}{\delta q(x)} (\mathcal{F}(q)+a \mathcal{G}(q))\Big)\\
=&\textrm{grad}_{\mathrm{T}}(\mathcal{F}+a\mathcal{G})(q).
\end{split}
\end{equation*}
Following Proposition \ref{prop5}, we finish the proof.
\end{proof}
We next formulate a duality theorem for transport Bregman divergences. Denote $C^{\infty}(\Omega)/\mathbb{R}$ be a smooth functional space, whose element is uniquely determined up to a constant shrift. Consider a functional $\mathcal{F}\colon \mathcal{P}(\Omega)\rightarrow\mathbb{R}$. 
Denote $\mathcal{\tilde F}\colon C^{\infty}(\Omega)/\mathbb{R}\rightarrow\mathbb{R}$, such that
\begin{equation}\label{def}
\mathcal{\tilde F}(\Phi)=\mathcal{F}(\nabla\Phi_\# q).
\end{equation}
\begin{theorem}[Transport duality]\label{thm}
Assume that functional $\mathcal{\tilde F}$ is convex w.r.t. $\Phi$ in $L^2$ space.
Denote $\mathcal{F}^*\colon C^{\infty}(\Omega)/\mathbb{R}\rightarrow \mathbb{R}$, such that
\begin{equation}\label{BD}
\mathcal{\tilde F}^*(\Phi^*)=\sup_{\Phi\in C^{\infty}(\Omega)/\mathbb{R}}~\int_\Omega\Phi(x)\Phi^*(x)dx-\mathcal{\tilde F}(\Phi).
\end{equation}
Then the following relations hold. 
\begin{itemize}
\item[(i)]
\begin{equation*}
\nabla_x\frac{\delta}{\delta\Phi^*(x)}\mathcal{\tilde F^*}(\frac{\delta}{\delta\Phi}\mathcal{\tilde F}(\Phi))=\nabla_x\Phi(x), 
\end{equation*}
for any $\Phi\in C^{\infty}(\Omega)/\mathbb{R}$. 
\item[(ii)] 
\begin{equation*}
\mathrm{D}_{\rt,\mathcal{F}}(p\|q)=\mathrm{D}_{\mathcal{ \mathcal{\tilde F}}}(\Phi_p\|\Phi_q),
\end{equation*}
where $\mathrm{D}_{\mathcal{\tilde F}}$ is a Bregman divergence in $L^2$ space satisfying 
\begin{equation*}
\mathrm{D}_{\mathcal{\tilde F}}(p\|q)= \mathcal{\tilde F}(\Phi_p)-\mathcal{\tilde F}(\Phi_q)-\int_\Omega (\Phi_p(x)-\Phi_q(x))\cdot\frac{\delta}{\delta\Phi_q(x)}\mathcal{\tilde F}(\Phi_q)dx.
\end{equation*}
\item[(iii)] 
\begin{equation*}
\mathrm{D}_{\mathcal{ \mathcal{\tilde F}}}(\Phi_p\|\Phi_q)=\mathrm{D}_{\mathcal{\tilde F}^*}(\Phi^*_q\|\Phi^*_p)=\mathrm{D}_{\rt,\mathcal{F}}(p\|q),
\end{equation*}
where $\Phi_p^*$, $\Phi_q^*\in C^{\infty}(\Omega)$ satisfy
\begin{equation*}
\Phi_p^*(x)=\frac{\delta}{\delta\Phi_p(x)}\mathcal{\tilde F}(\Phi_p),\quad \Phi_q^*(x)=\frac{\delta}{\delta\Phi_q(x)}\mathcal{\tilde F}(\Phi_q),
\end{equation*}
and $\mathrm{D}_{\mathcal{\tilde F^*}}$ is a Bregman divergence in $L^2$ space, such that
\begin{equation*}
\mathrm{D}_{\mathcal{\tilde F^*}}(\Phi_q^*\|\Phi_p^*)= \mathcal{\tilde F}^*(\Phi_q)-\mathcal{\tilde F}^*(\Phi_p)-\int_\Omega \big(\Phi^*_q(x)-\Phi^*_p(x)\big)\cdot\frac{\delta}{\delta\Phi^*_p(x)}\mathcal{\tilde F}^*(\Phi^*_p)dx.
\end{equation*}
\end{itemize}
\end{theorem}
\begin{proof}
We notice that using transport coordinates \eqref{nc}, the $L^2$--Wasserstein space is flat. And the proof here follows from the duality of transport coordinates in $L^2$ space. 

\noindent(i) Consider  
\begin{equation*}
\Phi=\arg\sup_{\Phi\in C^\infty(\Omega)/\mathbb{R}}~\int_\Omega\Phi(x)\Phi^*(x)dx-\mathcal{\tilde F}(\Phi).
\end{equation*}
Then 
\begin{equation*}
\Phi^*=\frac{\delta}{\delta \Phi}\mathcal{\tilde F}(\Phi).
\end{equation*}
Also, consider 
\begin{equation*}
\mathcal{\tilde F}(\Phi)=\mathcal{\tilde F}^{**}(\Phi)=\sup_{\Phi\in C^{\infty}(\Omega)/\mathbb{R}}~\int_\Omega\Phi(x)\Phi^*(x)dx-\mathcal{\tilde F^*}(\Phi^*).
\end{equation*}
Hence 
\begin{equation*}
\Phi=\frac{\delta}{\delta\Phi^*}\mathcal{\tilde F}^*(\Phi^*)=\frac{\delta}{\delta\Phi^*}\mathcal{\tilde F}^*(\frac{\delta}{\delta\Phi}\mathcal{\tilde F}(\Phi)),
\end{equation*}
where the equality in above formula holds up to a constant shrift. 

\noindent(ii) We next prove $\mathrm{D}_{\rt,\mathcal{F}}(p\|q)=\mathrm{D}_{\mathcal{\tilde F^*}}(\Phi_p\|\Phi_q)$. From the definition of functional $\mathcal{\tilde F}$, it suffices to prove that 
\begin{equation}\label{aa}
\frac{\delta}{\delta\Phi_q(x)}\mathcal{\tilde F}(\Phi_q)=-\nabla_x\cdot(q(x)\nabla_x\frac{\delta}{\delta q(x)}\mathcal{F}(q)).
\end{equation}
Here the proof follows \cite[Lemma 10.4.1]{AGS}. For the completeness of paper, we derive \eqref{aa} here. Consider a small perturbation $h\in C^{\infty}(\Omega)$ with $\epsilon\in\mathbb{R}$, then 
\begin{equation*}
\mathcal{\tilde F}(\Phi_q+\epsilon h)=\mathcal{F}((\nabla\Phi_q+\epsilon \nabla h)_\#q)=\mathcal{F}((\mathrm{id}+\epsilon \nabla h)_\#q),
\end{equation*}
where $\mathrm{id}(x)=x$ is an identity map. Standard calculations show that 
\begin{equation*}
\begin{split}
\frac{d}{d\epsilon}\mathcal{\tilde F}(\Phi_q+\epsilon h)|_{\epsilon=0}=&\int_\Omega (\nabla_x h(x),\nabla_x\frac{\delta}{\delta q(x)}\mathcal{F}(q)) q(x)dx\\
=&-\int_\Omega h(x) \nabla_x\cdot(q(x)\nabla_x\frac{\delta}{\delta q(x)}\mathcal{F}(q))dx,
\end{split}
\end{equation*}
which finishes the proof.

\noindent(iii) Given $p$, $q\in\Omega$, denote $\Phi^*_p=\frac{\delta}{\delta\Phi_p}\mathcal{\tilde F}(\Phi_p)$, $\Phi^*_q=\frac{\delta}{\delta\Phi_q}\mathcal{\tilde F}(\Phi_q)$. 
We only need to prove the duality of Bregman divergence $\mathrm{D}_{\mathcal{\tilde{F}^*}}(\Phi_p\|\Phi_q)$ in $L^2$ space. On the one hand, 
\begin{equation*}
\begin{split}
\mathrm{D}_{\mathcal{\tilde F}}(\Phi_p\|\Phi_q)=& \mathcal{\tilde F}(\Phi_p)-\mathcal{\tilde F}(\Phi_q)-\int_\Omega(\Phi_p(x)-\Phi_q(x))\cdot \frac{\delta}{\delta\Phi_q(x)}\mathcal{\tilde F}(\Phi_q) dx\\
=& \mathcal{\tilde F}(\Phi_p)+\Big(\int_\Omega \Phi_q(x)\cdot\frac{\delta}{\delta\Phi_q(x)}\mathcal{\tilde F}(\Phi_q)dx-\mathcal{\tilde F}(\Phi_q)\Big)-\int_\Omega \Phi_p(x)\cdot\frac{\delta}{\delta\Phi_q}\mathcal{\tilde F}(\Phi_q) dx\\
=& \mathcal{\tilde F}(\Phi_p)+\mathcal{\tilde F}^*(\Phi_q^*)-\int_\Omega \Phi_p(x)\Phi_q^*(x) dx,
\end{split}
\end{equation*}
where we apply the definition of conjugate functional $\mathcal{\tilde F}^*$ in $L^2$ space. 
On the other hand, 
\begin{equation*}
\begin{split}
\mathrm{D}_{\mathcal{\tilde F}^*}(\Phi^*_q\|\Phi^*_p)=& \mathcal{\tilde F^*}(\Phi_q^*)-\mathcal{\tilde F^*}(\Phi_p)-\int_\Omega(\Phi^*_q(x)-\Phi^*_p(x))\cdot\frac{\delta}{\delta\Phi^*_p(x)}\mathcal{\tilde F^*}(\Phi^*_p) dx\\
=& \mathcal{\tilde F^*}(\Phi^*_q)+\Big(\int_\Omega \Phi_p^*(x)\cdot\frac{\delta}{\delta\Phi_p^*(x)}\mathcal{\tilde F^*}(\Phi^*_p)dx-\mathcal{\tilde F}(\Phi^*_p)\Big)-\int_\Omega \Phi_q^*(x)\cdot\frac{\delta}{\delta\Phi^*_p}\mathcal{\tilde F^*}(\Phi^*_p) dx\\
=& \mathcal{\tilde F^*}(\Phi_q)+\mathcal{\tilde F}(\Phi_p)-\int_\Omega \Phi_p(x)\Phi_q^*(x) dx,
\end{split}
\end{equation*}
where we apply the definition of conjugate of conjugate functional $\mathcal{\tilde F}=\mathcal{\tilde F}^{**}$ in $L^2$ space. This finishes the proof. 
\end{proof}
\begin{remark}
We emphasize the fact that the Legendre duality in Theorem \ref{thm} represents the one in probability density space $\mathcal{P}(\Omega)$. It is different from the Legendre duality of the ground cost defined in sample space $\Omega$. 
\end{remark}
\section{Transport KL divergence}\label{section4}
In this section, we study the transport Bregman divergence of negative Boltzmann--Shannon entropy, namely transport KL divergence. 
\subsection{Review of KL divergence}
We first review some facts about KL divergence, which is defined by
\begin{equation*}
\mathrm{D}_{\textrm{KL}}(p\| q)=\int_\Omega p(x)\log\frac{p(x)}{q(x)}dx.
\end{equation*}
It is a Bregman divergence in $L^2$ space of negative $\mathcal{H}$, where $\mathcal{H}$ represents the Boltzmann--Shannon entropy defined by
\begin{equation*}
\mathcal{H}(p)=-\int_\Omega p(x)\log p(x)dx.
\end{equation*}
 There are several properties of KL divergence, which are useful in Bayesian sampling and AI inference problems. Firstly, define the cross entropy by 
\begin{equation*}
\mathcal{H}_q(p)=-\int_\Omega p(x)\log q(x)dx.
\end{equation*}
Hence KL divergence can be formulated by
\begin{equation*}
\begin{split}
\mathrm{D}_{\textrm{KL}}(p\|q)=&-\int_\Omega p(x)\log q(x)dx+\int_\Omega p(x)\log p(x)dx\\
=&\mathcal{H}_q(p)-\mathcal{H}(p).
\end{split}
\end{equation*}
Secondly, KL divergence has many desired properties, such as nonnegativity, separability etc. Lastly, KL divergence can be used to generate other divergences functionals. Notice that KL divergence is non-symmetric w.r.t. $p$, $q$. In practice, one can define a symmetrized divergence. One typical example is the Jenson-Shannon divergence defined by
\begin{equation*}
\mathrm{D}_{\mathrm{JS}}(p\|q)=\frac{1}{2}\mathrm{D}_{\mathrm{KL}}(p\|r)+\frac{1}{2}\mathrm{D}_{\mathrm{KL}}(q\|r), 
\end{equation*}
where $r=\frac{p+q}{2}$ is the geodesic midpoint (Barycenter) in $L^2$ space.
\subsection{Transport KL divergence}
We are now ready to derive an analog of KL divergence in $L^2$--Wasserstein space. We consider a transport Bregman divergence of (negative) Boltzmann-Shannon entropy as in Example \ref{ent} (i).
\begin{definition}[Transport KL divergence]\label{DTKL}
Define $\mathrm{D}_{\rtKL}\colon \mathcal{P}(\Omega)\times \mathcal{P}(\Omega)\rightarrow \mathbb{R}$ by 
\begin{equation*}
\begin{split}
\mathrm{D}_{\rtKL}(p\|q)=&\int_\Omega \Big(\Delta_x\Phi_p(x)-\log\mathrm{det}(\nabla^2_x\Phi_p(x))-d\Big)q(x)dx,
\end{split}
\end{equation*}
where $\nabla_x\Phi_p$ is the differemorphism map from $q$ to $p$, such that 
\begin{equation*}
(\nabla_x\Phi_p)_\# q= p.
\end{equation*}
Denote $\pi=(\mathrm{id}\times \nabla\Phi_p)_\#q$, then 
\begin{equation*}
\mathrm{D}_{\rtKL}(p\|q)=\int_\Omega \Big(p(x)\log p(x)-q(x)\log q(x)\Big)dx+\int_\Omega \Big(y\cdot\nabla_x\frac{\pi(x,y)}{q(x)}-d\Big)q(x)dx.
\end{equation*}
If $d=1$, then
 \begin{equation*}
\mathrm{D}_{\rtKL}(p\|q)=\int_0^1\Big(\frac{\nabla_xF_p^{-1}(x)}{\nabla_xF_q^{-1}(x)}-\log \frac{\nabla_xF_p^{-1}(x)}{\nabla_xF_q^{-1}(x)}-1\Big)dx.
 \end{equation*}
We call $\mathrm{D}_{\rtKL}$ the transport KL divergence. 
\end{definition}
We study several properties of transport KL divergence. We first derive an analog of cross entropy in $L^2$--Wasserstein space below.
\begin{definition}[Transport cross entropy]
Define $\mathcal{H}_{\rt,q}\colon \mathcal{P}(\Omega)\rightarrow \mathbb{R}$ by
\begin{equation*}
\begin{split}
\mathcal{H}_{\rt,q}(p)=&\int_\Omega \Delta_x\Phi_p(x) q(x)dx-\int_\Omega q(x)\log q(x)dx-d,
\end{split}
\end{equation*}
where $\nabla_x{\Phi_p}_\#q=p$. We call $\mathrm{H}_{\rt}$ the transport cross entropy. If $d=1$, then 
 \begin{equation*}
\mathcal{H}_{\rt,q}(p)=\int_0^1\frac{\nabla_xF_p^{-1}(x)}{\nabla_xF_q^{-1}(x)}dx-\int_\Omega q(x)\log q(x)dx-1.
 \end{equation*}
\end{definition}
We next demonstrate an equality for transport KL divergence. 
\begin{proposition}
\begin{equation*}
\mathrm{D}_{\rtKL}(p\|q)=\mathcal{H}_{\rt,q}(p)-\mathcal{H}(p).
\end{equation*}
\end{proposition}
\begin{proof}
The proof follows from Definition \ref{DTKL}. As in Example \ref{ent} (i), from ${\nabla_x\Phi_p}_\#q=p$, i.e.,
\begin{equation}\label{MA}
p(\nabla_x\Phi_p(x))\mathrm{det}(\nabla_x^2\Phi_p(x))=q(x),
\end{equation}
we have
\begin{equation*}
\begin{split}
\mathcal{H}(q)=&\int_\Omega q(x)\log q(x)dx\\
=&\int_\Omega p(\nabla_x\Phi_p(x))\mathrm{det}(\nabla_x^2\Phi_p(x))\log \Big(p(\nabla_x\Phi_p(x))\mathrm{det}(\nabla_x^2\Phi_p(x))\Big)dx\\
=&\int_\Omega p(\nabla_x\Phi_p(x))\mathrm{det}(\nabla_x^2\Phi_p(x))\Big(\log p(\nabla_x\Phi_p(x))+\log \mathrm{det}(\nabla_x^2\Phi_p(x))\Big)dx\\
=&\int_\Omega q(x)\log\mathrm{det}(\nabla_x^2\Phi_p(x))dx+\int_\Omega p(y)\log p(y)dy\\
=&\int_\Omega q(x)\log\mathrm{det}(\nabla_x^2\Phi_p(x))dx-\mathcal{H}(p),
\end{split}
\end{equation*}
where we use the fact that $y=\nabla_x\Phi_p(x)$ in the last equality. Hence 
\begin{equation*}
\begin{split}
\mathrm{D}_{\rtKL}(p\|q)
=&\int_\Omega \Big(\Delta_x\Phi_p(x)-\log\mathrm{det}(\nabla^2_x\Phi_p(x))-d\Big)q(x)dx\\
=&\int_\Omega p(x)\log p(x)dx-\int_\Omega q(x)\log q(x)dx+\int_\Omega \Delta_x\Phi_p(x)q(x)dx-d\\
=&-\mathcal{H}(p)+\mathcal{H}_{\rt,q}(p),
\end{split}
\end{equation*}
which finishes the proof.
\end{proof}
We last derive several properties of transport KL divergence.
\begin{theorem}
The transport KL divergence has the following properties.
\begin{itemize}
\item[(i)] Nonnegativity: For any $p$, $q\in\mathcal{P}(\Omega)$, then 
\begin{equation*}
\mathrm{D}_{\rtKL}(p\|q)\geq 0.
\end{equation*}
And 
\begin{equation*}
\mathrm{D}_{\rtKL}(p\|q)=0\quad\textrm{iff}\quad p(x+c)=q(x), \quad\textrm{for any $x\in\Omega$,}
\end{equation*}
where $c\in\mathbb{R}^d$ is a constant vector. 
\item[(ii)] Separability: The transport KL divergence is additive for independent distributions. Suppose $(p_1, p_2), (q_1, q_2)\in\mathcal{P}(\Omega)\times \mathcal{P}(\Omega)$ are independent distributions with joint distributions 
\begin{equation*}
p(x,y)=p_1(x)p_2(y),\qquad q(x,y)=q_1(x)q_2(y).
\end{equation*}
Then 
\begin{equation*}
\mathrm{D}_{\rtKL}(p\|q)=\mathrm{D}_{\rtKL}(p_1\|q_1)+\mathrm{D}_{\rtKL}(p_2\|q_2).
\end{equation*}
\item[(iii)] Transport Hessian information metric: Denote $\sigma=-\nabla\cdot(q\nabla\Phi)\in T_q\mathcal{P}(\Omega)$ and $\epsilon\in\mathbb{R}$. Consider a Taylor expansion by
\begin{equation*}
\mathrm{D}_{\rtKL}((\mathrm{id}+\epsilon \nabla\Phi )_\# q \|q)=\frac{\epsilon^2}{2}\mathrm{Hess}_{\mathrm{T}}\mathcal{H}(q)(\sigma,\sigma)+o(\epsilon^2).
\end{equation*}
Here $\mathrm{Hess}_{\mathrm{T}}\mathcal{H}$ is the Hessian operator of (negative) Boltzmann--Shanon entropy $\mathcal{H}(q)$ on $L^2$--Wasserstein space \cite{Villani2009_optimal}, known as the transport Hessian information metric \cite{LiHess}. In details, 
\begin{equation*}
\mathrm{Hess}_{\mathrm{T}}\mathcal{H}(q)(\sigma, \sigma)=\int_\Omega \mathrm{tr}\Big(\nabla_x^2\Phi(x), \nabla_x^2\Phi(x)\Big)q(x)dx,
\end{equation*}
where $\Phi, \sigma \in C^{\infty}(\Omega)$ satisfy the elliptical equation 
\begin{equation*}
\sigma(x)=-\nabla_x\cdot\Big(q(x)\nabla_x\Phi(x)\Big).
\end{equation*}
\item[(iv)] Transport convexity: Denote $p_1=(\nabla_x\Phi_{p_1})\# q$ and $p_2=(\nabla_x\Phi_{p_2})_\# q$ with
\begin{equation*}
p_\lambda=\Big(\lambda \nabla_x\Phi_{p_1}+(1-\lambda)\nabla_x\Phi_{p_2}\Big)_\# q. 
\end{equation*}
Then for any $\lambda\in [0,1]$, we have
\begin{equation*}
\mathrm{D}_{\rtKL}(p_\lambda\|q)\leq \lambda \mathrm{D}_{\rtKL}(p_1\|q)+(1-\lambda)\mathrm{D}_{\rtKL}(p_2\|q).
\end{equation*}
\end{itemize}
\end{theorem}
\begin{proof}
Here (iii) follows from Proposition \ref{TBP}. We only need to prove (i), (ii) and (iv). 

\noindent(i) Denote the eigendecomposition of a symmetric matrix function by 
\begin{equation*}
\nabla^2\Phi_p(x)=V^{\ts}\Lambda V, 
\end{equation*}
where $\Lambda=\textrm{diag}(\lambda_i(x))_{1\leq i\leq d}$ is a positive eigenvalue matrix and $V$ is an orthogonal eigenvector matrix. Then
\begin{equation*}
\begin{split}
\mathrm{D}_{\rtKL}(p\|q)=&\sum_{i=1}^d\int_\Omega (\lambda_i(x)-\log\lambda_i(x)-1)q(x)dx\geq 0. 
\end{split}
\end{equation*}
In the above proof, we use the fact that $\lambda_i-\log \lambda_i-1\geq 0$, where the equality holds when $\lambda_i=1$. Hence $\mathrm{D}_{\rtKL}(p\|q)=0$ implies the fact 
that $\lambda_i(x)=1$ for $i=1,\cdots, d$, on the support of $q$. Thus $\nabla_x^2\Phi_p(x)=\mathbb{I}$. In other words, $\nabla_x\Phi_p(x)=x+c$, where $c$ is a constant vector in $\mathbb{R}^d$. From $\nabla_x{\Phi_p}_\#q=p$, 
i.e. $p(\nabla_x\Phi_p)\mathrm{det}(\nabla^2\Phi_p)=q$, we prove the result. 
 
\noindent(ii) Consider pushforward operators $\nabla_x\Phi_{p_1}$, $\nabla_x\Phi_{p_2}$ from $q_1$, $q_2$ to $p_1$, $p_2$, respectively. In other words, 
\begin{equation*}
\nabla_x\Phi_{p_1}(x)_\# q_1(x)=p_1(x),\qquad \nabla_y\Phi_{p_2}(y)_\# q_2(y)=p_2(y).
\end{equation*}
Notice 
\begin{equation*}
(\nabla_x\Phi_{p_1}(x), \nabla_y\Phi_{p_2}(y))_\# ( q_1(x)q_2(y))=p_1(x)p_2(y).
\end{equation*}
Denote $\Phi_p(x,y)=\Phi_{p_1}(x)+\Phi_{p_2}(y)$. Then $\nabla\Phi(x,y)=(\nabla_x\Phi_{p_1}(x), \nabla_y\Phi_{p_2}(y))$. Thus 
\begin{equation*}
\nabla \Phi(x,y)_\# q(x,y)=p(x,y).
\end{equation*}
We are now ready to check the separability property. Notice $\textrm{dim}(\Omega\times \Omega)=2d$, then
\begin{equation*}
\Delta \Phi_p(x,y)=\Delta_x\Phi_{p_1}(x)+\Delta_y\Phi_{p_2}(y),
\end{equation*}
and 
\begin{equation*}
\begin{split}
\mathrm{det}(\nabla^2\Phi_p(x,y))=&\mathrm{det}\begin{pmatrix} \nabla^2_{xx}\Phi_{p_1}(x)&0\\
0 & \nabla^2_{yy}\Phi_{p_2}(y)
\end{pmatrix}\\
=&\mathrm{det}(\nabla^2_{xx}\Phi_{p_1}(x))\cdot \mathrm{det}(\nabla^2_{yy}\Phi_{p_2}(y)).
\end{split}
\end{equation*}
Hence
\begin{equation*}
\begin{split}
\mathrm{D}_{\rtKL}(p\|q)=&\int_\Omega \int_\Omega\Big(\Delta\Phi_p(x,y)-\log\mathrm{det}(\nabla^2\Phi_p(x,y))-2d\Big)q(x,y)dxdy\\
=&\int_\Omega\int_\Omega \Big(\Delta_x\Phi_{p_1}(x)+\Delta_y\Phi_{p_2}(y)\\
&\hspace{1cm}-\log \mathrm{det}(\nabla_x^2\Phi_{p_1}(x))-\log\mathrm{det}(\nabla_y^2\Phi_{p_2}(y))-2d\Big) q_1(x)q_2(y)dxdy\\
=&\int_\Omega\int_\Omega \Big(\Delta_x\Phi_{p_1}(x)-\log \mathrm{det}(\nabla_x^2\Phi_{p_1}(x))-d\\
&\hspace{1cm}+\Delta_y\Phi_{p_2}(y)-\log\mathrm{det}(\nabla_y^2\Phi_{p_2}(y))-d)\Big) q_1(x)q_2(y)dxdy\\
=&\mathrm{D}_{\rtKL}(p_1\|q_1)+\mathrm{D}_{\rtKL}(p_2\|q_2).
\end{split}
\end{equation*}

\noindent(iv) Denote a matrix function $f\colon \mathbb{R}^{d\times d}\rightarrow \mathbb{R}$, such that 
\begin{equation*}
f(A)=\mathrm{tr}(A)-\log\mathrm{det}(A)-d. 
\end{equation*}
Then $f$ is convex w.r.t. the matrix variable $A$. In this notation, we have 
\begin{equation*}
\mathrm{D}_{\rtKL}(p\|q)=\int_\Omega f(\nabla_x^2\Phi_p(x))q(x)dx.
\end{equation*}
Thus 
\begin{equation*}
\begin{split}
\mathrm{D}_{\rtKL}(p_\lambda\|q)=&\int_\Omega f\Big(\lambda \nabla^2_x\Phi_{p_1}(x)+(1-\lambda)\nabla^2_x\Phi_{p_2}(x)\Big)q(x)dx\\
\leq &\int_\Omega\Big(\lambda f(\nabla^2_x\Phi_{p_1}(x))+(1-\lambda) f(\nabla^2_x\Phi_{p_2}(x))\Big)q(x) dx\\
=&\lambda\mathrm{D}_{\rtKL}(p_1\|q)+(1-\lambda)\mathrm{D}_{\rtKL}(p_2\|q).
\end{split}
\end{equation*}
\end{proof}
\begin{remark}
We recall the fact that the convexity for classical KL divergence is on both terms of $p$, $q$ w.r.t. the $L^2$ metric. However, the classical KL divergence may not be convex in term of pushforward mapping function $\nabla_x\Phi_p$. A known fact is that the KL divergence in term of mapping depends on the convexity of negative log target distribution on sample space. This is different from the ones in transport KL divergence. Here the convexity w.r.t pushforward map $\nabla_x\Phi_p$ holds based on the definition of transport KL divergence. 
See detailed comparisons in appendix subsection \ref{sec52}.
\end{remark}
\begin{remark}
We notice that the transport convexity in (iv) can be different from the displacement convexity. We only show the transport convexity of transport KL divergence for a fixed reference density $q$. This fact is similar to the generalized convexity defined in \cite[Definition 9.2.4]{AGS}.
\end{remark}
\begin{remark}
We remark that the Hessian operator of (negative) Boltzmann--Shannon entropy in $L^2$--Wasserstein space connects Fisher-Rao metric, Gamma calculus and Ricci curvature on a sample space; see \cite{LiG, LiHess, Villani2009_optimal}. This paper proposes to construct Bregman divergences by using this transport information Hessian metric. 
\end{remark}
\subsection{Examples in Gaussian distributions}
Suppose that $p_X$, $p_Y$ are two Gaussian distributions with zero means in $\mathbb{R}^d$, such that
\begin{equation}\label{Gaussian}
p_X(x)=\frac{1}{\sqrt{(2\pi)^d\mathrm{det}(\Sigma_X)}}e^{-\frac{1}{2}x^{\ts}\Sigma_X^{-1}x},\quad p_Y(x)=\frac{1}{\sqrt{(2\pi)^d\mathrm{det}(\Sigma_Y)}}e^{-\frac{1}{2}x^{\ts}\Sigma_Y^{-1}x},
\end{equation}
where $\Sigma_X$, $\Sigma_Y\in\mathbb{R}^{d\times d}$ are symmetric positive definite matrices.
\begin{proposition}\label{prop18}
The transport KL divergence in Gaussian family \eqref{Gaussian} forms
\begin{equation}\label{TKL}
\begin{split}
\mathrm{D}_{\mathrm{TKL}}(p_X\|p_Y)=&\frac{1}{2}\log\frac{\mathrm{det}(\Sigma_Y)}{\mathrm{det}(\Sigma_X)}+\mathrm{tr}\Big(\Sigma_X^{\frac{1}{2}}\Big(\Sigma_X^{\frac{1}{2}}\Sigma_Y\Sigma_X^{\frac{1}{2}}\Big)^{-\frac{1}{2}}\Sigma_X^{\frac{1}{2}}\Big)-d.
\end{split}
\end{equation}
When $\Sigma_X$, $\Sigma_Y$ commute, i.e. $\Sigma_X\Sigma_Y=\Sigma_Y\Sigma_X$, \eqref{TKL} can be simplified by
\begin{equation*}
\mathrm{D}_{\mathrm{TKL}}(p_X\|p_Y)=\frac{1}{2}\log\frac{\mathrm{det}(\Sigma_Y)}{\mathrm{det}(\Sigma_X)}+\mathrm{tr}\Big(\Sigma_X^{\frac{1}{2}}\Sigma_Y^{-\frac{1}{2}}\Big)-d.
\end{equation*}
\end{proposition}
\begin{proof}
The proof follows from the fact \cite{AT} that the transport map $\nabla\Phi_{p_X}$, with the property ${\nabla\Phi_{p_X}}_\#p_Y=p_X$, satisfies  
\begin{equation*}
\nabla_x\Phi_{p_X}(x)=\Sigma_X^{\frac{1}{2}}\Big(\Sigma_X^{\frac{1}{2}}\Sigma_Y\Sigma_X^{\frac{1}{2}}\Big)^{-\frac{1}{2}}\Sigma_X^{\frac{1}{2}}x. 
\end{equation*}
Then $\Delta_x\Phi_{p_X}(x)=\mathrm{tr}\Big(\Sigma_X^{\frac{1}{2}}\Big(\Sigma_X^{\frac{1}{2}}\Sigma_Y\Sigma_X^{\frac{1}{2}}\Big)^{-\frac{1}{2}}\Sigma_X^{\frac{1}{2}}\Big)$.
Combining these facts into Definition \ref{DTKL}, we derive formula \eqref{TKL}.
\end{proof}
\begin{example}
We remark that the transport KL divergence \eqref{TKL} is different from the classical KL divergence, where
\begin{equation*}
\begin{split}
\mathrm{D}_{\mathrm{KL}}(p_X\|p_Y)=&\int_\Omega p_X(x)\log\frac{p_X(x)}{p_Y(x)}dx\\
=&\frac{1}{2}\log\frac{\mathrm{det}(\Sigma_Y)}{\mathrm{det}(\Sigma_X)}+\frac{1}{2}\mathrm{tr}\Big(\Sigma_X \Sigma_Y^{-1}\Big)-\frac{d}{2}.
\end{split}
\end{equation*}
We compare both KL divergence and transport KL divergence in Figure \ref{figure}. 
 \begin{figure}[H]
    \includegraphics[scale=0.4]{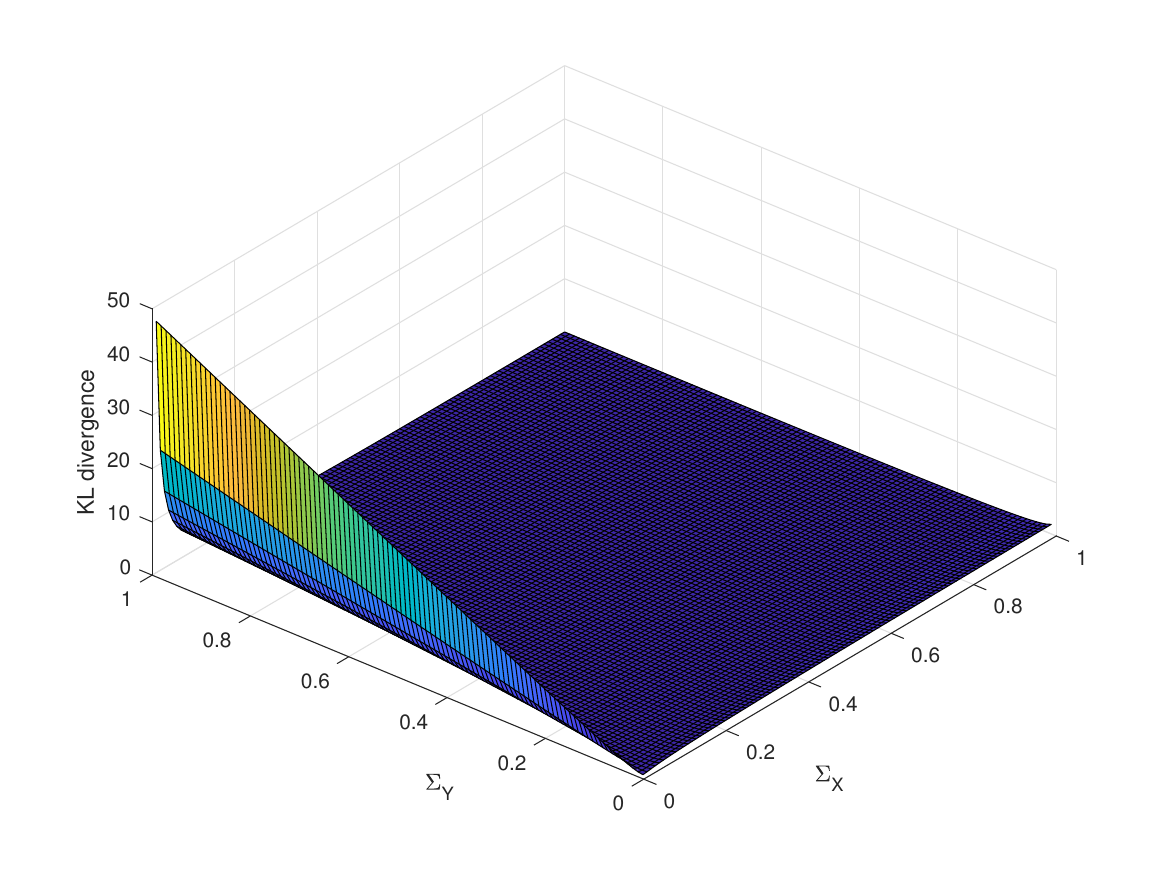}\includegraphics[scale=0.4]{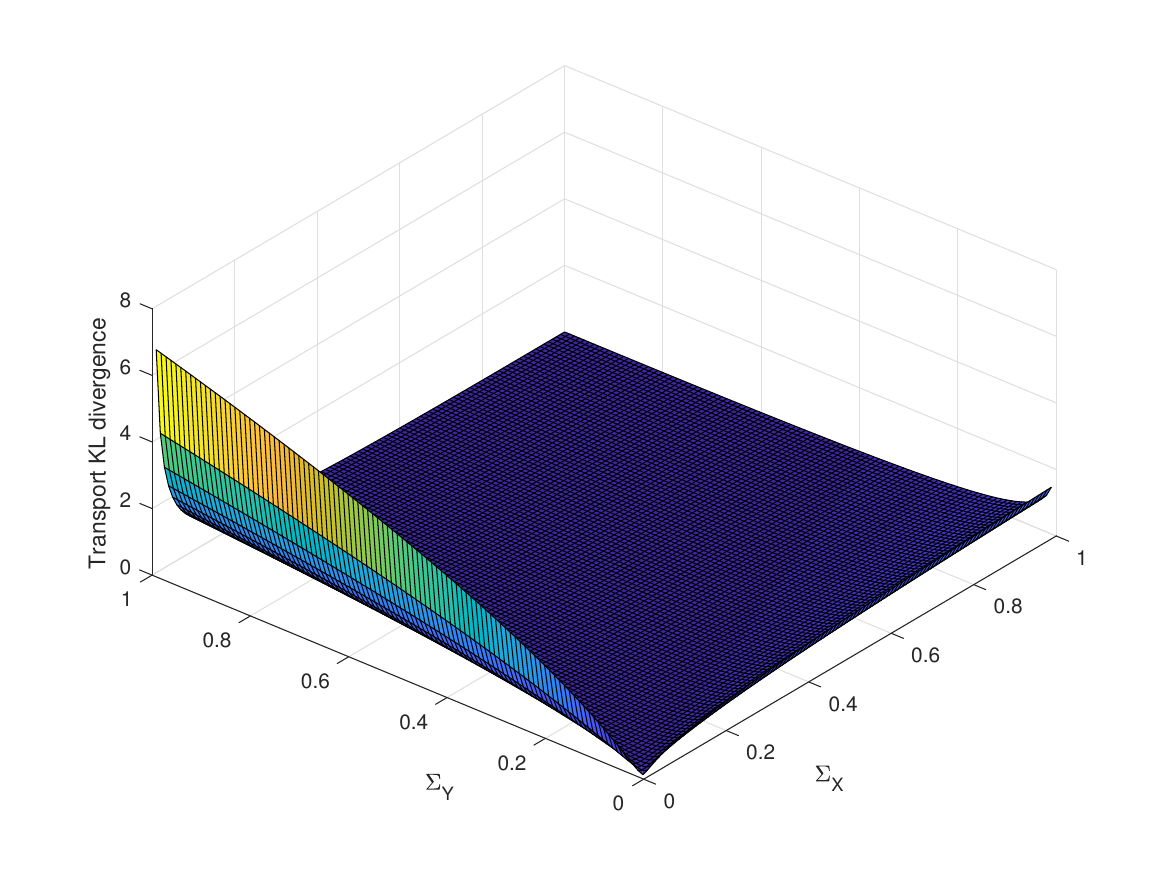}
    \caption{A comparison between KL divergence and transport KL divergence for one dimensional Gaussian distributions. Left represents the KL divergence. Right represents the transport KL divergence.
    }
    \label{figure}
\end{figure}
\end{example}
\subsection{Transport Jenson--Shannon divergence}
We also define a symmetrized KL divergence in $L^2$--Wasserstein space. We call it the transport Jenson--Shannon divergence. 
\begin{definition}[Transport Jensen--Shannon divergence]
Define $\mathrm{D}_{\rtJS}\colon \mathcal{P}(\Omega)\times \mathcal{P}(\Omega)\rightarrow \mathbb{R}$ by
\begin{equation*}
\mathrm{D}_{\rtJS}(p\|q)=\frac{1}{2}\mathrm{D}_{\rtKL}(p\|r)+\frac{1}{2}\mathrm{D}_{\rtKL}(q\|r),\end{equation*}
where $r\in\mathcal{P}(\Omega)$ is the geodesic midpoint (Barycenter) between $p$ and $q$ in $L^2$--Wasserstein space, i.e.
\begin{equation*}
r=\Big(\frac{1}{2}\big(\nabla_x\Phi_p+\nabla_x\Phi_q\big)\Big)_\# q.
\end{equation*}
\end{definition}
We present several closed form solutions for transport Jenson--Shanon divergence. 
\begin{proposition}\label{prop11}
The transport Jenson--Shanon divergence in one dimensional sample space satisfies 
\begin{equation*}
\mathrm{D}_{\rtJS}(p\|q)=-\frac{1}{2}\int_0^1\log\frac{\nabla_xF_p^{-1}(x)\cdot \nabla_xF_q^{-1}(x)}{\frac{1}{4}(\nabla_xF_p^{-1}(x)+\nabla_xF_q^{-1}(x))^2} dx.
\end{equation*}
\end{proposition}
\begin{proposition}\label{prop12}
The transport Jenson--Shanon divergence in Gaussian family \eqref{Gaussian} satisfies
\begin{equation}\label{TJS}
\begin{split}
&\mathrm{D}_{\rtJS}(p_X\|p_Y)\\
=&-\frac{1}{2}\log\frac{\mathrm{det}(\Sigma_X^\frac{1}{2})\mathrm{det}(\Sigma_Y^\frac{1}{2})}{\mathrm{det}(\Sigma_Z)}+\frac{1}{2}\mathrm{tr}\Big(\Sigma_X^{\frac{1}{2}}\Big(\Sigma_X^{\frac{1}{2}}\Sigma_Z\Sigma_X^{\frac{1}{2}}\Big)^{-\frac{1}{2}}\Sigma_X^{\frac{1}{2}}+ \Sigma_Y^{\frac{1}{2}}\Big(\Sigma_Y^{\frac{1}{2}}\Sigma_Z\Sigma_Y^{\frac{1}{2}}\Big)^{-\frac{1}{2}}\Sigma_Y^{\frac{1}{2}} \Big)-d,
\end{split}
\end{equation}
where 
\begin{equation*}
\Sigma_Z=\frac{1}{4}(\mathbb{I}+\mathbb{T})\Sigma_Y(\mathbb{I}+\mathbb{T}),\quad\textrm{with}\quad \mathbb{T}=\Sigma_X^{\frac{1}{2}}\Big(\Sigma_X^{\frac{1}{2}}\Sigma_Y\Sigma_X^{\frac{1}{2}}\Big)^{-\frac{1}{2}}\Sigma_X^{\frac{1}{2}}.
\end{equation*}
In particular, if $\Sigma_X$, $\Sigma_Y$ commute, then \eqref{TJS} can be simplified by
\begin{equation*}
\mathrm{D}_{\rtJS}(p_X\|p_Y)=-\frac{1}{2}\log\frac{\mathrm{det}(\Sigma_X^\frac{1}{2})\mathrm{det}(\Sigma_Y^{\frac{1}{2}})}{\mathrm{det}\Big(\frac{1}{4}(\Sigma^{\frac{1}{2}}_X+\Sigma^{\frac{1}{2}}_Y)^2\Big)}.
\end{equation*}
\end{proposition}
\begin{proof}[Proof of Proposition \ref{prop11} and \ref{prop12}]
If $d=1$, from \cite[Theorem 6.0.2]{AGS}, the geodesic $p_t(x)\in \mathcal{P}(\Omega)$, $t\in [0,1]$, in $L^2$--Wasserstein space connecting $q$ to $p$ satisfies 
\begin{equation*}
\nabla_x F^{-1}_{p_t}(x)=t\nabla_x F^{-1}_p(x)+(1-t) \nabla_xF^{-1}_q(x).
\end{equation*}
The geodesic midpoint, i.e. $r=p_{\frac{1}{2}}$, satisfies
\begin{equation*}
\nabla_x F^{-1}_r(x)=\nabla_x F^{-1}_{p_{\frac{1}{2}}}(x)=\frac{1}{2}\Big({\nabla_x F_p^{-1}(x)}+{\nabla_x F_q^{-1}(x)}\Big).
\end{equation*}
From example \ref{ent} (i), we prove the result. Similarly, in Gaussian family, the results follow the proof in Proposition \ref{prop18}. Here the geodesic midpoint $p_{\frac{1}{2}}=\mathcal{N}(0, \Sigma_Z)$ in $L^2$--Wasserstein space satisfies  
\begin{equation*}
\Sigma_Z=(\frac{1}{2}(\mathbb{I}+\mathbb{T}))^{\ts}\Sigma_Y(\frac{1}{2}(\mathbb{I}+\mathbb{T}))=\frac{1}{4}(\mathbb{I}+\mathbb{T})\Sigma_Y(\mathbb{I}+\mathbb{T}). 
\end{equation*}
In particular, if $\Sigma_X$, $\Sigma_Y$ commute, then 
\begin{equation*}
\Sigma_Z^{\frac{1}{2}}=\frac{1}{2}\Sigma_X^{\frac{1}{2}}+\frac{1}{2}\Sigma_Y^{\frac{1}{2}}.
\end{equation*}
From Proposition \ref{prop18} and formula \eqref{TJS}, we derive the result. 
\end{proof}
 \begin{figure}
 \centering
    \includegraphics[scale=0.4]{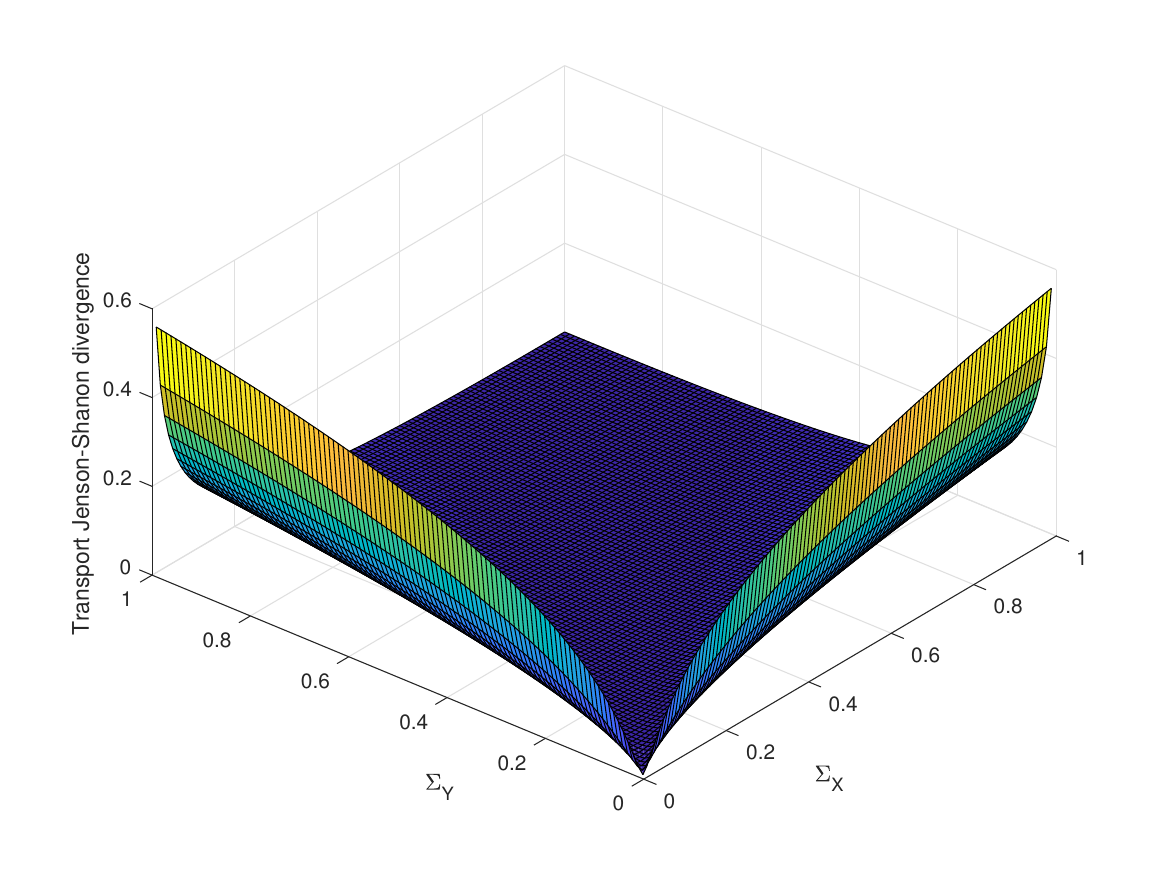}
    \caption{Illustration of transport Jensen--Shannon divergences for one dimensional Gaussian distributions.}
    \label{figure2}
\end{figure}
\begin{remark}
A known fact is that the classical Jensen--Shannon divergence between Gaussian distributions does not have a closed-form solution \cite{Nielsen2}. Interestingly, the transport Jensen--Shannon entropy has a closed form solution. And the geodesic midpoint in $L^2$--Wasserstein space provides the other way to construct symmetric divergence functionals. 
\end{remark}
\section{Discussion}
In this paper, we formulate Bregman divergences in $L^2$--Wasserstein space, namely transport Bregman divergences. They are generalizations of the $L^2$--Wasserstein distance. We also derive the transport KL divergence by a transport Bregman divergence of negative Boltzmann--Shannon entropy. We remark that the transport KL divergence is an Itakura--Saito type divergence in one-dimensional sample space. We also propose a symmetrized generalization of transport KL divergence.  These transport divergences have shown convexity properties in terms of pushforward mapping functions. We expect that transport Bregman divergences will be useful in AI inference and optimization problems.

\newpage
\section*{Appendix}
In this appendix, we first present all derivation proofs in this paper. We next summarize several formulations of transport Bregman divergences.  
We last compare the convexity difference between KL divergence and transport KL divergence.
\subsection{Proofs in section \ref{sec3}}
We present all derivations of transport information Bregman divergences. 
\begin{proof}[Proof of Proposition \ref{prop7}]
Our proof follows from Definition \ref{def4}. Notice that $T(x)=\nabla_x\Phi_p(x)$, $x=\nabla_x\Phi_q(x)$. 

(i) Since $\frac{\delta}{\delta q(x)}\mathcal{V}(q)=V(x)$, then 
\begin{equation*}
\begin{split}
\mathrm{D}_{\rt,\mathcal{V}}(p\|q)=&\int_\Omega V(x)p(x)dx-\int_\Omega V(x)q(x)dx-\int_\Omega (T(x)-x, \nabla_xV(x))q(x)dx\\
=&\int_\Omega \Big(V(T(x))-V(x)-(T(x)-x, \nabla_xV(x))\Big)q(x)dx\\
=&\int_\Omega \mathrm{D}_{V}(T(x)\|x)q(x)dx.
\end{split}
\end{equation*}

(ii) Since $\frac{\delta}{\delta q(x)}\mathcal{W}(q)=\int_\Omega W(x,\tilde x)q(\tilde x)d\tilde x$, then
\begin{equation*}
\begin{split}
\mathrm{D}_{\rt,\mathcal{W}}(p\|q)=&\frac{1}{2}\int_\Omega\int_\Omega W(x,\tilde x)p(x)p(\tilde x)dxd\tilde x-\frac{1}{2}\int_\Omega\int_\Omega W(x,\tilde x)q(x)q(\tilde x)dxd\tilde x\\
&-\int_\Omega\int_\Omega (T(x)-x, \nabla_xW(x,\tilde x))q(x)q(\tilde x)dxd\tilde x\\
=&\frac{1}{2}\int_\Omega\int_\Omega W(T(x),T(\tilde x))q(x)q(\tilde x)dxd\tilde x-\frac{1}{2}\int_\Omega\int_\Omega W(x,\tilde x)q(x)q(\tilde x)dxd\tilde x\\
&-\frac{1}{2}\int_\Omega\int_\Omega (T(x)-T(\tilde x)-(x-\tilde x), \nabla_xW(x,\tilde x))q(x)q(\tilde x)dxd\tilde x\\
=&\frac{1}{2}\int_\Omega\int_\Omega \Big(\tilde W(T(x)-T(\tilde x))-\tilde W(x-\tilde x)\\
&\hspace{1.4cm}-(T(x)-T(\tilde x)-(x-\tilde x), \nabla\tilde W(x-\tilde x))\Big)q(x)q(\tilde x)dxd\tilde x\\
=&\frac{1}{2}\int_\Omega\int_\Omega \mathrm{D}_{\tilde W}(T(x)-T(\tilde x)\|x-\tilde x)q(x)q(\tilde x)dxd\tilde x.
\end{split}
\end{equation*}
In above derivations, the second equality uses the pushforward relation 
\begin{equation*}
\int_\Omega\int_\Omega W(x,\tilde x)p(x)p(\tilde x)dx d\tilde x=\int_\Omega\int_\Omega W(T(x), T(\tilde x))q(x)q(\tilde x)dxd\tilde x,
\end{equation*}
and applies the equality that 
\begin{equation*}
\begin{split}
&\int_\Omega\int_\Omega (T(x)-x, \nabla_xW(x,\tilde x))q(x)q(\tilde x)dxd\tilde x\\
=&\frac{1}{2}\int_\Omega\int_\Omega (T(x)-x, \nabla_xW(x,\tilde x))q(x)q(\tilde x)dxd\tilde x+\frac{1}{2}\int_\Omega\int_\Omega (T(\tilde x)-\tilde x, \nabla W(\tilde x,x))q(x)q(\tilde x)dxd\tilde x\\
=&\frac{1}{2}\int_\Omega\int_\Omega (T(x)-x-(T(\tilde x)-\tilde x), \nabla_xW(x,\tilde x))q(x)q(\tilde x)dxd\tilde x,
\end{split}
\end{equation*}
where $\nabla_xW(x,\tilde x)=-\nabla_{\tilde x}W(x,\tilde x)$.

(iii) In this case, $\frac{\delta}{\delta q(x)}\mathcal{U}(q)=U'(q(x))$. Then 
\begin{equation}\label{u1}
\mathrm{D}_{\rt,\mathcal{U}}(p\|q)=\mathcal{U}(p)-\mathcal{U}(q)-\int_\Omega \big(T(x)-x, \nabla_xU'(q(x))\big)q(x)dx.
\end{equation}
Denote $y=T(x)$, then 
\begin{equation}\label{u2}
\begin{split}
\mathcal{U}(p)=&\int_\Omega U(p(y))dy\\
=&\int_\Omega U(p(T(x))) \mathrm{det}(\nabla_xT(x))dx\\
=&\int_\Omega U(\frac{q(x)}{\mathrm{det}(\nabla_xT(x))})\mathrm{det}(\nabla_xT(x))dx,
\end{split}
\end{equation}
where we use the fact 
\begin{equation*}
p(T(x))\mathrm{det}(\nabla_xT(x))=q(x).
\end{equation*}
In addition, 
\begin{equation}\label{u3}
\begin{split}
&-\int_\Omega (T(x)-x, \nabla_xU'(q(x)))q(x)dx\\
=&\int_\Omega \nabla_x\cdot\big(q(x)(T(x)-x)\big) U'(q(x))dx \\
=&\int_\Omega \Big(\big(\nabla_xq(x), T(x)-x\big)U'(q(x))+\nabla_x\cdot(T(x)-x)U'(q(x))q(x)\Big)dx\\
=&\int_\Omega \Big(\nabla_xU(q(x)), T(x)-x\Big)dx+\int_\Omega \Big(\nabla_x\cdot(T(x)-x)U'(q(x))q(x)\Big)dx\\
=&\int_\Omega \nabla_x\cdot(T(x)-x)\big[U'(q(x))q(x)-U(q(x))\big]dx.
\end{split}
\end{equation}
Substituting \eqref{u2} and \eqref{u3} into \eqref{u1}, and formulating the results in term of a matrix divergence form, we derive the result. 
\end{proof}
\begin{proof}[Proof of Proposition \ref{prop8}]
The proof follows the relation between the joint density and the mapping. (i) For the linear energy, we have 
\begin{equation*}
\begin{split}
\mathrm{D}_{\rt,\mathcal{V}}(p\|q)=&\int_\Omega \mathrm{D}_V(T(x)\|x)q(x)dx\\
=&\int_\Omega \int_\Omega \mathrm{D}_V(y\|x)\pi(x,y)dxdy. 
\end{split}
\end{equation*}
\noindent(ii) For the interaction energy, we have 
\begin{equation*}
\begin{split}
\mathrm{D}_{\rt,\mathcal{W}}(p\|q)=&\frac{1}{2}\int_\Omega\int_\Omega \mathrm{D}_{\tilde W}(T(x)-T(\tilde x)\|x-\tilde x)q(x)q(\tilde x)dxd\tilde x\\
=&\frac{1}{2}\int_{\Omega}\int_{\Omega}\int_{\Omega}\int_{\Omega} \mathrm{D}_{\tilde W}(y-\tilde y\|x-\tilde x)\pi(x,y)\pi(\tilde x,\tilde y)dxdyd\tilde xd\tilde y. 
\end{split}
\end{equation*}
\noindent(iii) For the entropy, we have 
\begin{equation*}
\begin{split}
\mathrm{D}_{\rt,\mathcal{U}}(p\|q)=&\int_\Omega\Big(U(p(x))-U(q(x))+(\nabla_x\cdot T(x)-d)\bar U(q(x))\Big)dx\\
=&\quad\int_\Omega U(\int_\Omega \pi(x,y)dx)dy-\int_\Omega U(\int_\Omega\pi(x,y)dy)dx\\
&+\int_\Omega (\int_\Omega y\cdot \nabla_x\frac{\pi(x,y)}{q(x)}dy-d)\bar U(\int_\Omega \pi(x,z)dz)dx,
\end{split}
\end{equation*}
which finishes the proof. 
\end{proof}
\begin{proof}[Proof of Proposition \ref{prop9}]
The proof applies the pushforward relation:
\begin{equation*}
p(T(x))|\nabla_x T(x)|=q(x). 
\end{equation*}
Notice that the optimal map $T(x)$ in one dimensional spatial domain is monotonically increasing, i.e. $\nabla_x T(x)\geq 0$. Then 
\begin{equation*}
p(T(x))\nabla_xT(x)=q(x),
\end{equation*}
i.e. 
\begin{equation*}
dF_p(T(x))=dF_q(x), 
\end{equation*}
where $F_p$, $F_q$ are cumulative functions of $p$, $q$, respectively.  One can solve the above equation by 
\begin{equation*}
T(x)=F_p^{-1}(F_q(x)).
\end{equation*}
Denote $y=F_q(x)$. Then $x=F_q^{-1}(y)$ and 
\begin{equation*}
q(x)=\frac{dy}{dx}=\frac{1}{\frac{dx}{dy}}=\frac{1}{\nabla_yF_q^{-1}(y)}.
\end{equation*}
We are now ready to derive closed form solutions for \eqref{TB}.

\noindent(i) For \eqref{linear}, we have 
\begin{equation*}
\begin{split}
\mathrm{D}_{\rt,\mathcal{V}}(p\|q)=&\int_\Omega \mathrm{D}_{V}(T(x)\|x)q(x)dx\\
=&\int_\Omega \mathrm{D}_{V}(F^{-1}_p(F_q(x))\|x)dF_q(x)\\
=&\int_0^1 \mathrm{D}_{V}(F^{-1}_p(y)\|F^{-1}_q(y))dy.
\end{split}
\end{equation*}
(ii) For \eqref{interact}, we have 
\begin{equation*}
\begin{split}
\mathrm{D}_{\rt,\mathcal{W}}(p\|q)=&\frac{1}{2}\int_\Omega\int_\Omega \mathrm{D}_{\tilde W}(T(x)-T(\tilde x)\|x-\tilde x)q(x)q(\tilde x)dxd\tilde x\\
=&\frac{1}{2}\int_\Omega\int_\Omega \mathrm{D}_{\tilde W}(F^{-1}_p(F_q(x))-F^{-1}_p(F_q(\tilde x))\|x-\tilde x)dF_q(x)dF_q(\tilde x)\\
=&\frac{1}{2}\int_0^1\int_0^1 \mathrm{D}_{\tilde W}(F^{-1}_p(y)-F^{-1}_p(\tilde y)\|F_q^{-1}(y)-F_q^{-1}(\tilde y))dyd\tilde y.
\end{split}
\end{equation*}
(iii) For \eqref{entropy}, we have
\begin{equation*}
\begin{split}
\mathrm{D}_{\rt,\mathcal{U}}(p\|q)=&\int_\Omega \Big\{U(\frac{q(x)}{\nabla_x T(x)})\nabla_x T(x)-U(q(x))+\nabla_x\cdot(T(x)-x)\big(q(x)U'(q(x))-U(q(x))\big)\Big\}dx\\
=&\int_\Omega \Big\{U(\frac{q(x)}{\nabla_x F_p^{-1}(F_q(x))})\nabla_x F_p^{-1}(F_q(x))-U(q(x))\\
&\quad +(\nabla_xF_p^{-1}(F_q(x))-1)\big(q(x)U'(q(x))-U(q(x))\big)\Big\}dx\\
=&\int_0^1 \Big\{U(\frac{1}{\nabla_yF_q^{-1}(y)\nabla_x F_p^{-1}(y)})\nabla_x F_p^{-1}(y)-U(\frac{1}{\nabla_yF_q^{-1}(y)})\\
&\quad +(\nabla_xF_p^{-1}(y)-1)\big(\frac{1}{\nabla_yF_q^{-1}(y)}U'(\frac{1}{\nabla_yF_q^{-1}(y)})-U(\frac{1}{\nabla_yF_q^{-1}(y)})\big)\Big\}{\nabla_yF_q^{-1}(y)}dy\\
=&\int_0^1 \Big\{U(\frac{1}{\nabla_yF_p^{-1}(y)})\nabla_y F_p^{-1}(y)-U(\frac{1}{\nabla_yF_q^{-1}(y)})\nabla_yF_q^{-1}(y)\\
&\quad +(\nabla_yF_p^{-1}(y)-\nabla_yF_q^{-1}(y))\big(\frac{1}{\nabla_yF_q^{-1}(y)}U'(\frac{1}{\nabla_yF_q^{-1}(y)})-U(\frac{1}{\nabla_yF_q^{-1}(y)})\big)\Big\}dy\\
=&\int_0^1 \Big\{ \tilde U(\nabla_y F_p^{-1}(y))-\tilde U(\nabla_yF_q^{-1}(y)) -\big(\nabla_yF_p^{-1}(y)-\nabla_yF_q^{-1}(y)\big)\tilde U'(\nabla_y F_q^{-1}(y))\Big\}dy\\
=&\int_0^1 \mathrm{D}_{\tilde U}(\nabla_yF_p^{-1}(y)\|\nabla_yF_q^{-1}(y))dy,
\end{split}
\end{equation*}
where the fourth equality applies the chain rule that 
\begin{equation*}
y=F_q(x),\quad \nabla_xF_p^{-1}(y)=\nabla_yF_p^{-1}(y)\cdot\frac{dy}{dx}=\frac{\nabla_yF_p^{-1}(y)}{\nabla_yF_q^{-1}(y)},
\end{equation*}
and the last equality applies the fact that
\begin{equation*}
\tilde U(z)=zU(\frac{1}{z}),\quad \tilde U'(z)=U(\frac{1}{z})-\frac{1}{z}U'(\frac{1}{z}).
\end{equation*}
This finishes the proof.
\end{proof}
\subsection{Closed form formulas}
In this subsection, we summarize the derived transport divergence functions as follows. Given $p$, $q\in\mathcal{P}(\Omega)$, denote the convex function $\Phi_p\colon \Omega\rightarrow\mathbb{R}$, such that ${(\nabla_x\Phi_p)}_\#q=p$, i.e. 
\begin{equation*}
p(\nabla_x\Phi_p(x))\mathrm{det}(\nabla_x^2\Phi_p(x))=q(x).
\end{equation*}
Define the transport Bregman divergence of functional $\mathcal{F}(p)$ by
\begin{equation*}
\mathrm{D}_{\rt,\mathcal{F}}(p\|q)=\mathcal{F}(p)-\mathcal{F}(q)-\int_\Omega \Big(\nabla_x\frac{\delta}{\delta q(x)}\mathcal{F}(q), \nabla_x\Phi_p(x)-x\Big)q(x)dx.
\end{equation*}
Several examples of transport Bregman divergences are given below. 
\begin{itemize}
\item[(i)] Linear energy: If $\mathcal{V}(p)=\int_\Omega V(x)p(x)dx$, then 
\begin{equation*}
\begin{split}
\mathrm{D}_{\rt,\mathcal{V}}(p\|q)=&\int_\Omega \mathrm{D}_{V}(\nabla_x\Phi_p(x)\|x)q(x)dx,
\end{split}
\end{equation*}
where $\mathrm{D}_{V}$ is a Euclidean Bregman divergence of $V$ defined by
\begin{equation*}
\mathrm{D}_V(z_1\|z_2)=V(z_1)-V(z_2)-\nabla V(z_2)\cdot (z_1-z_2),\quad\textrm{for any $z_1,z_2\in\Omega$}.
\end{equation*}
\item[(ii)] Interaction energy: If $\mathcal{W}(p)=\frac{1}{2}\int_\Omega\int_\Omega \tilde W(x-\tilde x)p(x) p(\tilde x)dxd\tilde x$, then
\begin{equation*}
\begin{split}
\mathrm{D}_{\rt,\mathcal{W}}(p\|q)=&\frac{1}{2}\int_\Omega\int_\Omega \mathrm{D}_{\tilde W}\big(\nabla_x\Phi_p(x)-\nabla_{\tilde x}\Phi_p(\tilde x)\|x-\tilde x\big)q(x)q(\tilde x)dxd\tilde x,
\end{split}
\end{equation*}
where $\mathrm{D}_{\tilde W}$ is a Euclidean Bregman divergence of $\tilde W$ defined by
\begin{equation*}
\mathrm{D}_{\tilde W}(z_1\|z_2)=\tilde W(z_1)-\tilde W(z_2)-\nabla \tilde W(z_2)\cdot(z_1-z_2),\quad\textrm{for any $z_1,z_2\in\Omega$}.
\end{equation*}
\item[(iii)] Negative entropy: If $\mathcal{U}(p)=\int_\Omega U(p(x))dx$, then  
\begin{equation*}
\begin{split}
\mathrm{D}_{\rt,\mathcal{U}}(p\|q)
=\int_\Omega &\Big\{U(\frac{q(x)}{\mathrm{det}(\nabla^2_x \Phi_p(x))})\mathrm{det}(\nabla^2_x \Phi_p(x))-U(q(x))\\
&+(\Delta_x\Phi_p(x)-d)\Big(q(x)U'(q(x))-U(q(x))\Big)\Big\}dx.
\end{split}
\end{equation*}
\item[(iv)] Transport KL divergence: If $\mathcal{U}(p)=\int_\Omega p(x)\log p(x)dx$, then
\begin{equation*}
\mathrm{D}_{\mathrm{TKL}}(p\|q)=\int_\Omega \Big(\Delta_x\Phi_p(x)-\log\mathrm{det}(\nabla^2_x\Phi_p(x))-d\Big)q(x)dx.
\end{equation*}
\end{itemize}
Given one dimensional sample space, several closed formulas of transport divergences are given below.
\begin{itemize}
\item[(i)] Linear energy:
\begin{equation*}
\begin{split}
\mathrm{D}_{\rt,\mathcal{V}}(p\|q)=\int_0^1 \mathrm{D}_{V}(F^{-1}_p(x)\|F^{-1}_q(x))dx.
\end{split}
\end{equation*}
\item[(ii)] Interaction energy:
\begin{equation*}
\begin{split}
\mathrm{D}_{\rt,\mathcal{W}}(p\|q)=&\frac{1}{2}\int_0^1\int_0^1 \mathrm{D}_{\tilde W}(F^{-1}_p(x)-F^{-1}_p(\tilde x)\|F^{-1}_q(x)-F^{-1}_q(\tilde x))dxd\tilde x.
\end{split}
\end{equation*}
\item[(iii)] Negative entropy:
\begin{equation*}
\begin{split}
\mathrm{D}_{\rt,\mathcal{U}}(p\|q)=\int_0^1 \mathrm{D}_{\tilde U}(\nabla_xF_p^{-1}(x)\|\nabla_xF_q^{-1}(x))dx,
\end{split}
\end{equation*}
where 
\begin{equation*}
\tilde U(z)=zU(\frac{1}{z}),
\end{equation*}
and $\mathrm{D}_{\tilde U}$ is a Euclidean Bregman divergence of $\tilde U$ defined by 
\begin{equation*}
\mathrm{D}_{\tilde U}(z_1\|z_2)=\tilde U(z_1)-\tilde U(z_2)-\nabla_z\tilde U(z_2)\cdot (z_1-z_2).
\end{equation*}
If $U(p)=\frac{p^2}{2}$, then 
\begin{equation*}
\mathrm{D}_{\rt,\mathcal{U}}(p\|q)=\frac{1}{2}\int_0^1 (\frac{1}{\nabla_xF_p^{-1}(x)}-\frac{1}{\nabla_xF_q^{-1}(x)})^2\cdot \nabla_xF^{-1}_p(x)dx.
\end{equation*}
\item[(iv)] Transport KL divergence:  
\begin{equation*}
\mathrm{D}_{\rtKL}(p\|q)=\int_0^1\Big(\frac{\nabla_xF_p^{-1}(x)}{\nabla_xF_q^{-1}(x)}-\log \frac{\nabla_xF_p^{-1}(x)}{\nabla_xF_q^{-1}(x)}-1\Big)dx.
\end{equation*}
\item[(v)] Transport Jenson--Shannon divergence: 
\begin{equation*}
\mathrm{D}_{\rtJS}(p\|q)=-\frac{1}{2}\int_{0}^1\log\frac{\nabla_xF_p^{-1}(x)\cdot \nabla_xF_q^{-1}(x)}{\frac{1}{4}(\nabla_xF_p^{-1}(x)+\nabla_xF_q^{-1}(x))^2} dx.
\end{equation*}
\end{itemize}
Given Gaussian distributions, several analytical formulas for transport information Bregman divergences are provided. Denote
\begin{equation*}
p_X(x)=\frac{1}{\sqrt{(2\pi)^d\mathrm{det}(\Sigma_X)}}e^{-\frac{1}{2}x^{\ts}\Sigma_X^{-1}x},\quad p_Y(x)=\frac{1}{\sqrt{(2\pi)^d\mathrm{det}(\Sigma_Y)}}e^{-\frac{1}{2}x^{\ts}\Sigma_Y^{-1}x}.
\end{equation*}
\begin{itemize}
\item[(i)] Transport KL divergence:
\begin{equation*}
\begin{split}
\mathrm{D}_{\mathrm{TKL}}(p_X\|p_Y)=\frac{1}{2}\log\frac{\mathrm{det}(\Sigma_Y)}{\mathrm{det}(\Sigma_X)}+\mathrm{tr}\Big(\Sigma_X^{\frac{1}{2}}\Big(\Sigma_X^{\frac{1}{2}}\Sigma_Y\Sigma_X^{\frac{1}{2}}\Big)^{-\frac{1}{2}}\Sigma_X^{\frac{1}{2}}\Big)-d.
\end{split}
\end{equation*}
If $\Sigma_X$, $\Sigma_Y$ commute, i.e. $\Sigma_X\Sigma_Y=\Sigma_Y\Sigma_X$, then 
\begin{equation*}
\mathrm{D}_{\mathrm{TKL}}(p_X\|p_Y)=\frac{1}{2}\log\frac{\mathrm{det}(\Sigma_Y)}{\mathrm{det}(\Sigma_X)}+\mathrm{tr}\Big(\Sigma_X^{\frac{1}{2}}\Sigma_Y^{-\frac{1}{2}}\Big)-d.
\end{equation*}
\item[(ii)] Transport Jenson-Shannon divergence: 
\begin{equation*}
\begin{split}
\mathrm{D}_{\rtJS}(p_X\|p_Y)=&-\frac{1}{4}\log\frac{\mathrm{det}(\Sigma_X)\mathrm{det}(\Sigma_Y)}{\mathrm{det}(\Sigma_Z)^2}\\
&+\frac{1}{2}\mathrm{tr}\Big(\Sigma_X^{\frac{1}{2}}\Big(\Sigma_X^{\frac{1}{2}}\Sigma_Z\Sigma_X^{\frac{1}{2}}\Big)^{-\frac{1}{2}}\Sigma_X^{\frac{1}{2}}+ \Sigma_Y^{\frac{1}{2}}\Big(\Sigma_Y^{\frac{1}{2}}\Sigma_Z\Sigma_Y^{\frac{1}{2}}\Big)^{-\frac{1}{2}}\Sigma_Y^{\frac{1}{2}} \Big)-d,
\end{split}
\end{equation*}
where 
\begin{equation*}
\Sigma_Z=\frac{1}{4}(\mathbb{I}+\mathbb{T})\Sigma_Y(\mathbb{I}+\mathbb{T}), \quad\textrm{and}\quad \mathbb{T}=\Sigma_X^{\frac{1}{2}}\Big(\Sigma_X^{\frac{1}{2}}\Sigma_Y\Sigma_X^{\frac{1}{2}}\Big)^{-\frac{1}{2}}\Sigma_X^{\frac{1}{2}}.
\end{equation*}
If $\Sigma_X$, $\Sigma_Y$ commute, then
\begin{equation*}
\mathrm{D}_{\rtJS}(p_X\|p_Y)=-\frac{1}{2}\log\frac{\mathrm{det}(\Sigma_X^\frac{1}{2})\mathrm{det}(\Sigma_Y^{\frac{1}{2}})}{\mathrm{det}\Big(\frac{1}{4}(\Sigma^{\frac{1}{2}}_X+\Sigma^{\frac{1}{2}}_Y)^2\Big)}.
\end{equation*}
\end{itemize}
\subsection{Comparisons between KL divergence and transport KL divergence}\label{sec52}
In this subsection, we compare KL divergence with transport KL divergence. This is an example to compare Bergman divergences in $L^2$ space and $L^2$--Wasserstein space. To do so, we first formulate the classical KL divergence in optimal transport coordinates. 
\begin{proposition}[KL divergence in transport coordinates \cite{AGS}]
KL divergence has the following formulations:
\begin{equation*}
\begin{split}
\mathrm{D}_{\mathrm{KL}}(p\|q)=\int_\Omega q(x)\log q(x)dx-\int_\Omega \log\mathrm{det}(\nabla_x^2\Phi_p(x))q(x)dx-\int_\Omega \log q(\nabla_x\Phi_p(x))q(x)dx,
\end{split}
\end{equation*}
where $\nabla_x\Phi_p\#q=p$. 
\end{proposition}
\begin{proof}
Denote $y=\nabla_x\Phi_p(x)$. From the Monge-Amper{\'e} equation, we have
\begin{equation*}
\begin{split}
\mathrm{D}_{\mathrm{KL}}(p\|q)=&\int_\Omega p(y)\log\frac{p(y)}{q(y)}dy\\
=&\int_\Omega p(\nabla_x\Phi_p(x))\log\frac{p(\nabla_x\Phi_p(x))}{q(\nabla_x\Phi_p(x))}\mathrm{det}(\nabla^2_x\Phi_p(x))dx\\
=&\int_\Omega q(x)\log\frac{q(x)}{\mathrm{det}(\nabla_x^2\Phi_p(x))q(\nabla_x\Phi_p(x))}dx\\
=&\int_\Omega q(x)\log q(x)dx-\int_\Omega q(x)\log\mathrm{det}(\nabla_x^2\Phi_p(x))dx-\int_\Omega q(x)\log q(\nabla_x\Phi_p(x))dx.
\end{split}
\end{equation*}
\end{proof}
We notice that the negative Boltzmann entropy term is convex w.r.t both $p$ and $\nabla_x\Phi_p(x)$, i.e. 
\begin{equation*}
\begin{split}
\mathcal{H}(p)=&-\int_\Omega p(x)\log p(x)dx=\int_\Omega \log\mathrm{det}(\nabla_x^2\Phi_p(x))q(x)dx.
\end{split}
\end{equation*}
This is true because $t\log t$ is convex w.r.t. $t\in\mathbb{R}_+$, and $-\log\mathrm{det}(A)$ is convex w.r.t matrix $A\in\mathbb{R}^{d\times d}$. Hence we know that the major difference between KL divergence and transport KL divergence is the corresponding cross entropy term. Here the cross entropy in $L^2$ space leads to 
\begin{equation*}
\begin{split}
\mathcal{H}_q(p)=&-\int_\Omega \log q(x) p(x)dx=-\int_\Omega \log q(\nabla_x\Phi_p(x))q(x)dx.
\end{split}
\end{equation*}
While the transport cross entropy in $L^2$--Wasserstein space forms 
\begin{equation*}
\mathcal{H}_{\mathrm{T},q}(p)=\int_\Omega \Delta_x\Phi_p(x)q(x)dx-C, 
\end{equation*}
where $C=d+\int_\Omega q(x)\log q(x)dx$. To summarize, KL divergence $\mathrm{D}_{\mathrm{KL}}(p\|q)=-\mathcal{H}(p)+\mathcal{H}_q(p)$
is convex w.r.t. mapping function $\nabla_x\Phi_p(x)$ if $-\log q(x)$ is convex in $\Omega$. While transport KL divergence 
$\mathrm{D}_{\rtKL}(p\|q)=-\mathcal{H}(p)+\mathcal{H}_{\mathrm{T}, q}(p)$ is always convex w.r.t. mapping function $\nabla_x\Phi_p(x)$. 
\end{document}